\documentclass[twoside,11pt]{article}

\usepackage{jmlr2e}

\usepackage{paralist, amsmath, amssymb, color, graphicx, hyperref}
\usepackage[margin=1in]{geometry}
\DeclareGraphicsExtensions{.jpg,.pdf,.eps}
\graphicspath{{./Figures/}}
\usepackage{amssymb}
\usepackage{setspace}
\usepackage{mathtools}
\usepackage{color}
\usepackage{cancel}
\usepackage{algorithm,algorithmic}
\usepackage{verbatim}
\usepackage{subcaption}

\usepackage{stmaryrd}

\usepackage{xfrac}

\usepackage{mathrsfs}
\usepackage{amscd}
\usepackage{dsfont}
\usepackage{tikz}

\usepackage{enumerate}

\newcommand{\vinayak}[1]{\textcolor{orange}{[Vinayak: #1}]}
\newcommand{\awan}[1]{\textcolor{red}{[Jordan: #1}]}

\newtheorem{thm}{Theorem}
\newtheorem{prop}[thm]{Proposition}
\newtheorem{cor}[thm]{Corollary}
\newtheorem{lem}[thm]{Lemma}

\newtheorem{defn}[thm]{Definition} 
\newtheorem{assumption}[thm]{Assumption}
\newtheorem{ex}[thm]{Example}
\newtheorem{rem}[thm]{Remark}

\usepackage{listings}
\usepackage{color} 
\definecolor{mygreen}{RGB}{28,172,0} 
\definecolor{mylilas}{RGB}{170,55,241}
\definecolor{mygray}{gray}{0.95}

\newcommand{\mscr}[1]{\mathscr{#1}}

\newcommand{\twid}[1]{\widetilde{#1}}

\newcommand{\RR}{\mathbb{R}}

\newcommand{\EE}{\mathbb{E}}

\newcommand{\ga}{\gamma}
\newcommand{\de}{\delta}
\newcommand{\ep}{\epsilon}

\DeclareMathOperator*{\argmax}{arg\,max}
\DeclareMathOperator*{\argmin}{arg\,min}

\usepackage{mathtools}

\newcommand{\defeq}{\vcentcolon=}

\newcommand{\red}{}
\newcommand{\yellow}{\color{yellow}}

\newcommand{\iid}{\overset{\text{iid}}{\sim}}

\newcommand{\g}{g}

\begin{document}

\title{Privacy-Aware Rejection Sampling}

\author{\name Jordan Awan \email jawan@purdue.edu \\
       \addr Department of Statistics\\
       Purdue University\\
       West Lafayette, IN 47907, USA
       \AND
       \name Vinayak Rao \email varao@purdue.edu \\
       \addr Department of Statistics\\
    Purdue University  \\
      West Lafayette, IN 47907, USA}
      

\editor{}
\firstpageno{1}

\maketitle

\begin{abstract}
    Differential privacy (DP) offers strong theoretical privacy guarantees, though implementations of DP mechanisms may be vulnerable to side-channel attacks, such as timing attacks. When sampling methods such as MCMC or rejection sampling are used to implement a mechanism, the runtime {\red can leak private information}. We characterize the additional privacy cost due to the runtime of a rejection sampler in terms of both $(\ep,\de)$-DP as well as $f$-DP. We also show that unless the acceptance probability is constant across databases, the runtime of a rejection sampler does not satisfy $\ep$-DP for any $\ep$. We show that there is a similar breakdown in privacy with adaptive rejection samplers. We propose three modifications to the rejection sampling algorithm, with varying assumptions, to protect against timing attacks by making the runtime independent of the data. The modification with the weakest assumptions is an approximate sampler, introducing a small increase in the privacy cost, whereas the other modifications give perfect samplers.  We also use our techniques to develop an adaptive rejection sampler for log-H\"{o}lder densities, which also has data-independent runtime. We give several examples of DP mechanisms that fit the assumptions of our methods and can thus be implemented using our samplers.
\end{abstract}

\section{Introduction}

As more data is collected, analyzed, and published by researchers, companies, and government agencies, concerns about the privacy of the participating individuals have become more prominent \citep{lane2014privacy}. While there have been many methods of statistical disclosure control to combat this problem \citep{hundepool2012statistical}, differential privacy (DP) \citep{dwork2006calibrating} has arisen as the state-of-the-art framework for privacy protection, and is currently being implemented by Google \citep{erlingsson2014rappor}, Apple \citep{tang2017privacy}, Microsoft \citep{ding2017collecting}, and the US Census \citep{abowd2018us}. Differential privacy is based on a notion of plausible deniability, and requires the introduction of additional noise, beyond sampling, into the analysis procedure. Given the output of a DP mechanism, an adversary cannot determine with high probability whether any particular individual participated in the dataset \citep{wasserman2010statistical}. 

Because of the formal nature of DP, implementations of the mechanisms must be very careful to prevent unintentional privacy leaks through side-channels. Side-channel attacks have been a long-standing problem in computer systems, and  may consist of the execution time, power consumption, or memory usage of the system, to name a few \citep{joy2011side,nilizadeh2019diffuzz}. With differential privacy, the system can be made black-box to remove some of these side-channels, but may still be susceptible to timing attacks. {\red Such a side-channel may be present if the DP mechanism is part of a query-response framework, where users submit queries and the curator replies with a DP response; in this model, the adversary may measure the time between submitting the query and receiving the answer, and use this information as part of his attack.} PINQ \citep{mcsherry2009privacy} and Airavat \citep{roy2010airavat} were two of the earliest DP implementations, but were shown by \citet{haeberlen2011differential} to be vulnerable to timing attacks. FUZZ \citep{haeberlen2011differential} and GUPT \citep{mohan2012gupt} avoid timing attacks by working with simple queries for which the worst-case computational time can be determined. This solution works for simple DP tasks, but is nontrivial for complex DP mechanisms.

One of the most common and powerful DP mechanisms is the exponential mechanism \citep{mcsherry2007mechanism} which results in an unnormalized density of the form $\exp(g_D(x))$ that must be sampled from, {\red where $g_D$ is some function that depends on the database $D$}. The exponential mechanism has been widely used to tackle problems such as principal component analysis \citep{chaudhuri2013near,kapralov2013differentially,awan2019benefits},  $K$-means clustering, \citep{feldman2009private}, convex optimization \citep{bassily2014private,bassily2014privateRevisited}, robust regression \citep{asi2020near}, linear and quantile regression \citep{reimherr2019kng}, synthetic data \citep{snoke2018pmse}, and Bayesian data analysis \citep{wang2015privacy,minami2016differential,zhang2016differential,dimitrakakis2017differential} to name a few.

A challenge however is that for functions $g_D(x)$ encountered in practice, the unnormalized density $\exp(g_D(x))$ is often difficult to sample from.
In statistics and machine learning, there are many computational techniques to produce either exact or approximate samples from such distributions, including Markov chain Monte Carlo (MCMC), rejection sampling, and approximate Bayesian computing. However, there are two sources of privacy leaks when using these computational sampling methods: 1) when using approximate samplers, the resulting sample does not exactly follow the target distribution, with the error in the approximation resulting in an increased privacy risk,  2) with either an approximate or exact sampler, if the runtime of the algorithm depends on the database, then this side-channel {\red may leak private information} \citep{haeberlen2011differential}.

 We will consider the runtime of the algorithm as an additional output accessible to an adversary, and we will require that both the official output and the runtime jointly satisfy differential privacy. As \citet{haeberlen2011differential} point out, the simplest solution is to make the runtime independent of the dataset. In this paper we propose different modifications, under different assumptions, which produce rejection samplers with data-independent runtime, and are thus immune to timing attacks. 
 
 ~\\
 {\noindent \bf Contributions } 
First, we quantify the privacy risk of rejection and adaptive rejection sampling without any privacy-preserving modifications. As a properly implemented rejection sampler results in samples with distribution equal to the target, the only privacy concern is the runtime, which varies for different databases. We characterize the privacy risk {\red due to the runtime} of a simple rejection sampler in terms of both $(\ep,\de)$-DP and $f$-DP \citep{dong2022gaussian}. We also show that {\red the runtime of} a simple rejection sampler does not satisfy $\ep$-DP for any finite $\ep$ unless the acceptance rate is constant across databases. We similarly show that {\red the runtime of} an adaptive rejection sampler does not satisfy $\ep$-DP unless acceptance probabilities across databases converge in terms of a certain series. 

Given the increased privacy risk due to the runtime, we propose several modifications to rejection samplers, which make the runtime independent of the database: 1) choose the number of iterations to run the sampler ahead of time, based on a lower bound on the acceptance probability, 2) introduce an additive wait-time based on a worst-case dataset, 3) use squeeze functions to add an implicit wait-time. 
We also propose an adaptive rejection sampler with data-independent runtime, which can be applied to any log-H\"{o}lder density. The adaptive sampler is a modification of the (nearly) minimax optimal sampler from \citet{achddou2019minimax}, using the technique of squeeze functions. Finally, we give examples of the exponential mechanism which satisfy the assumptions of our methods.

 ~\\
{\noindent \bf Related work }
%
%
{\red Often side-channels are handled using more relaxed metrics than DP, such as min-entropy \citep{smith2009foundations}. However, the point of view of this paper is that if the dataset in question is judged to  require the protection of differential privacy, then we must ensure that the \emph{all channels} are protected in the DP framework. Thus, while for other applications it may be appropriate to use a weaker protection for side channels, in DP applications, the runtime must also satisfy DP. See \citet{haeberlen2011differential} for a similar discussion.}

Besides timing side-channels, there are other notable side-channel attacks that have been effective against DP implementations. \citet{haeberlen2011differential} showed that when the privacy budget is chosen based on the database, that the budget is another side-channel. {\red \citet{wagh2018differentially} consider the privacy cost of RAM access, and propose a differential privacy regime to formally protect the RAM access. }\citet{dodis2012differential} and \citet{garfinkel2020randomness} explore the concerns of using pseudo-random number generators in the implementation of DP systems. \citet{mironov2012significance} showed that when implementing DP mechanisms with floating point arithmetic, privacy can be arbitrarily compromised by the artifacts in the least significant bit. \citet{ilvento2020implementing} provide an implementation of the exponential mechanism on finite state spaces that is immune to the floating point attacks, but which is admitted to be susceptible to timing attacks.

A different approach to sampling the exponential mechanism is using MCMC techniques, and there have been some prior works characterizing the additional privacy cost of these approximate samplers. 
Usually convergence of MCMC methods is characterized in terms of total variation distance, and \citet{minami2016differential} showed that these guarantees can be imported to produce approximate DP samples with an increased `delta' in a fixed number of iterations. \citet{ganesh2020faster} expanded upon the results of \citet{vempala2019rapid} to show that Langevin MCMC converges in R\'{e}nyi divergence, which allows for the quantification of the privacy loss by sampling in terms of R\'{e}nyi DP. R\'{e}nyi divergences are much stronger than total variation, and have been used in various definitions of DP \citep{mironov2017Renyi,bun2016concentrated,bun2018composable}. \citet{minami2016differential} also study Langevin MCMC, but characterize the privacy cost in terms of $(\ep,\de)$-DP. {\red \citet{seeman2021exact} develop an exact sampler for the exponential mechanism based on an MCMC procedure with artificial atoms, however, they acknowledge that their approach does not protect against timing side-channels. }
To our knowledge, there has been no prior work quantifying the privacy risk of rejection sampling, or proposing rejection samplers with data-independent runtime.

\section{Background and notation}
In this section, we review the necessary background on differential privacy and rejection sampling. We also set the notation for the rest of the paper. 

Let $X$ and $Y$ be random variables on a measurable space $(\mscr Y,\mscr F)$, with corresponding probability measures $\mu_X$ and $\mu_Y$. The \emph{max-divergence} of $Y$ with respect to $X$ is {\red $D_\infty(Y||X) = \sup_{B\in \mscr F} \log \left(\frac{\mu_Y(B)}{\mu_X(B)}\right).$} 
If $\mu_X$ dominates $\mu_Y$, then $D_{\infty}(Y||X) = \sup_{y\in \mscr Y} \log \frac{d\mu_Y}{d\mu_X}(y)$, where $\frac{d\mu_Y}{d\mu_X}$ is the Radon-Nikodym derivative of $\mu_Y$ with respect to $\mu_X$. The \emph{symmetric max-divergence} is $D^S_{\infty}(X,Y)\defeq \max\{D_{\infty}(X||Y),D_\infty(Y||X)\}$.

For a distribution $M$, we typically write $\twid \pi(x)$ for an unnormalized density of $M$, $\pi(x) = \twid \pi(x)/\int \twid\pi(x) \ dx$, and $\g=\log(\twid \pi)$ (equivalently, $\twid \pi(x) = \exp(\g(x))$. We write $U(x)$ to denote a density that upper bounds $\twid \pi$ as $\twid \pi(x)\leq c_U U(x)$ for some constant $c_U$. Similarly, we write $L(x)$ for a density that lower bounds $\twid \pi$ as $c_L L(x)\leq \twid \pi(x)$ {\red for a constant $c_L$}. In rejection sampling, $U$ is called the \emph{proposal distribution}, and $L$ is the \emph{squeeze function}.

\subsection{Differential privacy}
Differential privacy (DP), introduced in \citet{dwork2006calibrating}, is a framework to characterize the privacy risk of a given algorithm, and offers techniques to design mechanisms which limit privacy loss. DP methods require the introduction of additional randomness, beyond sampling, in order to offer a notion of plausible deniability. Given the output of a DP mechanism, it is difficult for an adversary to determine whether a particular individual participated in the dataset or not. While an idealized algorithm may be proven to be differentially private, to characterize the actual privacy cost of a given implementation, one must consider all side-channels such as the runtime as part of the DP output \citep{haeberlen2011differential}.

\begin{defn}[Privacy Mechanism]
Given a metric space $(\mscr D,d)$, which represents the set of possible databases, a set of probability measures $\{M_D\mid D\in \mscr D\}$ on a common space $\mscr Y$ is called a \emph{privacy mechanism}.
\end{defn}
The space $\mscr D$ represents the space of possible databases, and it is common to take $\mscr D = \mscr X^n$ for some set $\mscr X$, with $\mscr X$ representing the possible contributions of one individual in the database.
In that case, the metric $d$ is often chosen to be the Hamming distance, so that $d(D,D')\leq 1$ represents that $D$ and $D'$ are {\em adjacent} databases, differing in only one individual's contribution.

When implementing a privacy mechanism, we publish one sample from $M_D$, which satisfies some form of privacy. 

\begin{defn}[$(\ep,\de)$-DP:\citealp{dwork2006calibrating}]
Given a metric space $(\mscr D,d)$, ${\red \ep\geq 0}$ and ${\red \de\in[0,1]}$, a privacy mechanism $\{M_D\}$ on the space $\mscr Y$ satisfies $(\ep,\de)$-differential privacy if for all measurable sets $B\in \mscr Y$ and all $d(D,D')\leq 1$,
\[M_D(B)\leq \exp(\ep) M_{D'}(B)+\de.\]
\end{defn}

The values $\ep$ and $\de$ are called the privacy parameters, which capture the privacy risk for the given mechanism. Smaller values of $\ep$ and $\de$ give stronger privacy guarantees. Typically, $\ep$ is chosen to be a small constant such as $1$ or $.1$, whereas $\de\ll 1/n$. In the case where $\de=0$, we call $(\ep,0)$-DP ``pure differential privacy,'' and write $\ep$-DP. A mechanism satisfying $\ep$-DP is equivalent to requiring that  $D^S_\infty(M_D||{\red M_{D'}})\leq \ep$ for all $d(D,D')\leq 1$, where $D^S_\infty$ is the symmetric max-divergence. 

While we phrase most of our results in terms of $(\ep,\de)$-DP, another useful formulation of DP is $f$-DP \citep{dong2022gaussian}, which is expressed in terms of hypothesis tests. $f$-DP is based on bounding the receiver-operator curve (ROC) or tradeoff function when testing between two adjacent databases, given the output of a mechanism. For two probability distributions $P$ and $Q$, the \emph{tradeoff function} is the smallest type-II error as a function of the type-I error. Formally, the \emph{tradeoff function} for $P$ and $Q$ is $T(P,Q):[0,1]\rightarrow [0,1]$, which is defined as $T(P,Q)(\alpha) = {\red \inf_\phi} \{1-\EE_Q (\phi) \mid \EE_P(\phi)\leq \alpha\}$, where the infinimum is over all possible tests $\phi$. Being equivalent to ROC, the tradeoff function captures the difficulty of distinguishing between $P$ and $Q$. A function $f:[0,1]\rightarrow [0,1]$ is a tradeoff function if and only if $f$ is convex, continuous, decreasing, and $f(x) \leq 1-x$ for all $x \in [0,1]$ \citep[Proposition 1]{dong2022gaussian}.

\begin{defn}[$f$-DP: \citealp{dong2022gaussian}]\label{def:fDP}
 Let $f$ be a tradeoff function. A mechanism $M$ on the metric space $(\mscr D,d)$ is $f$-DP if \[T(M_D,M_{D'})(\alpha)\geq f(\alpha)\quad \forall \alpha \in [0,1],\] for all $D,D'\in \mscr D$ such that $d(D,D')\leq 1$.
\end{defn}

See Figure \ref{fig:tradeoffLEFT} for examples of tradeoff functions which do and do not satisfy $f$-DP for a particular $f$. Without loss of generality we can assume that $f$ is symmetric: $f(\alpha) = f^{-1}(\alpha)$, where $f^{-1}(\alpha) = \inf \{t\in [0,1]\mid f(t)\leq \alpha\}$. This is due to  \citet[Proposition 2]{dong2022gaussian}, which states that for a given $f$ and a  mechanism $M$ that is $f$-DP, there exists a symmetric $f^*\geq f$ such that $M$ is $f^*$-DP. 

  \begin{figure}
      \centering
      \begin{subfigure}[t]{.48\textwidth}
      \includegraphics[width=\linewidth]{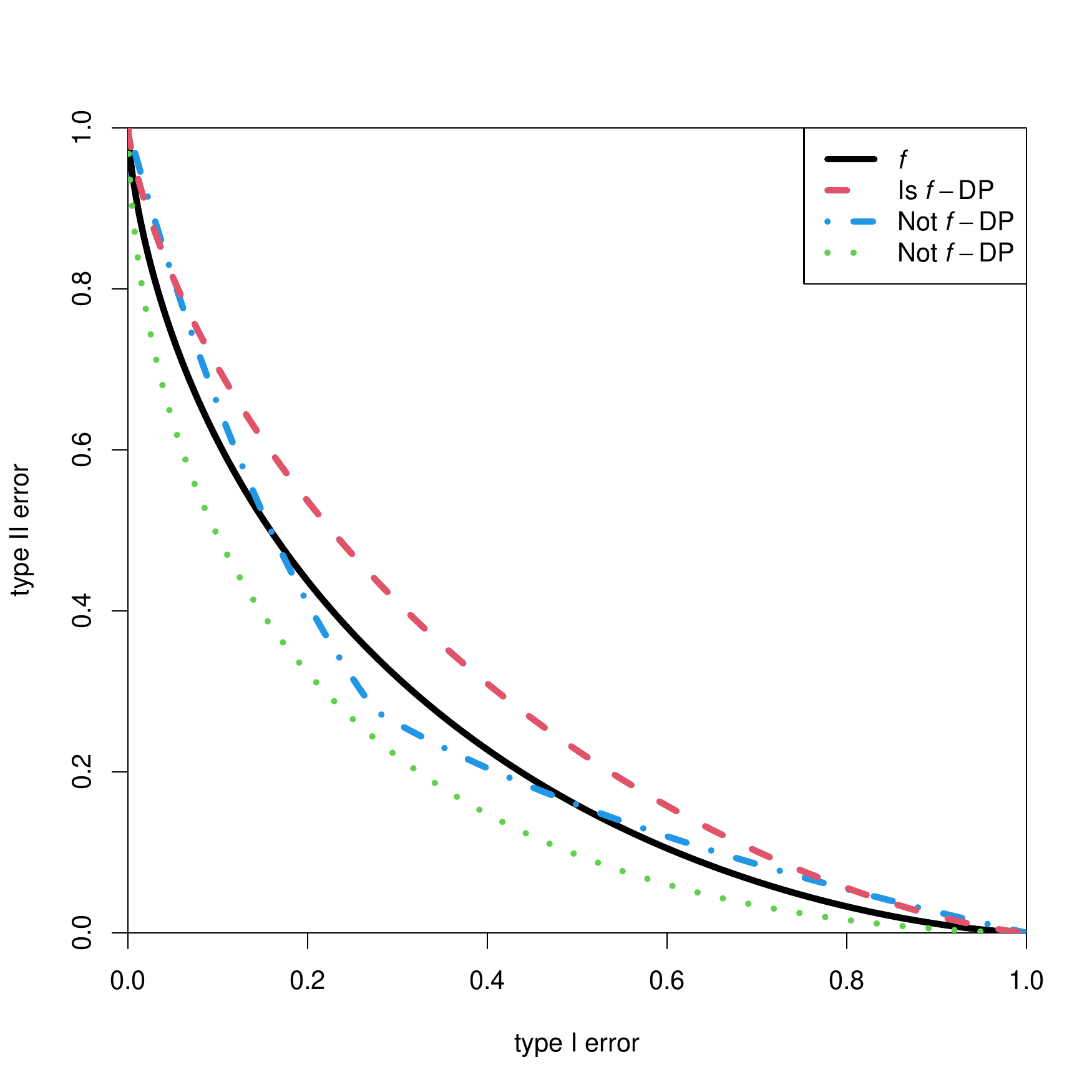}
      \caption{A plot of three examples of $T(M(D),M(D'))$. Only the red, dashed tradeoff curve satisfies $f$-DP.}
            \label{fig:tradeoffLEFT}
      \end{subfigure}
      \begin{subfigure}[t]{.48\textwidth}
      \includegraphics[width=\linewidth]{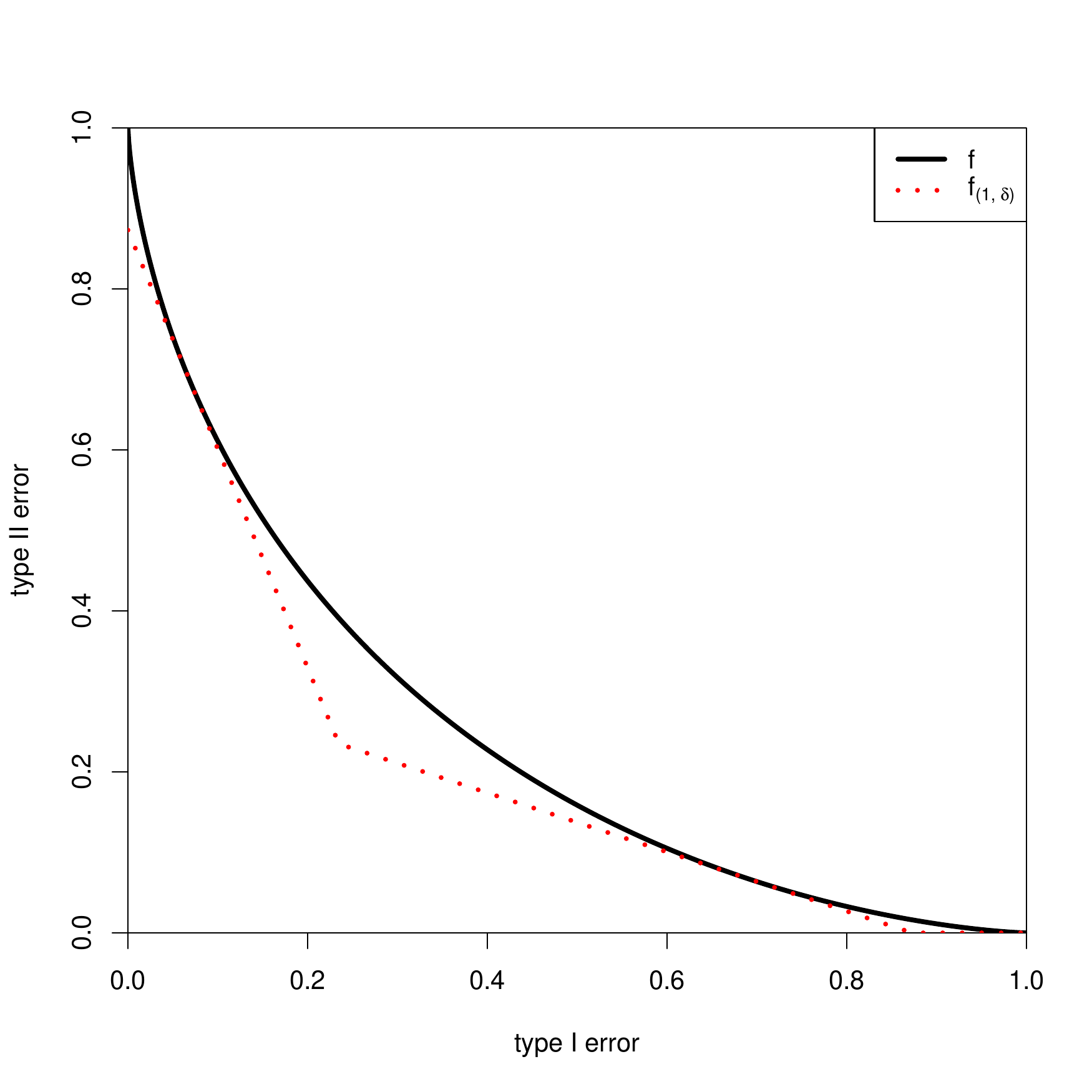}
      \caption{A tradeoff function, as well as its conversion to $(1,\de)$-DP, where $\de\approx0.127$.}
      \label{fig:f_to_DP}
      \end{subfigure}
      \caption{Examples of tradeoff functions, and the relation between $f$-DP and $(\ep,\de)$-DP. }
      \label{fig:tradeoff}
  \end{figure}


It turns out that $(\ep,\de)$-DP is a special case of $f$-DP, where $f$ is taken to be a particular piecewise linear function.
Specifically, let ${\red \ep\geq 0}$ and ${\red \de\in[0,1]}$, and define $f_{\ep,\de}(\alpha) = \max\{0,1-\de-\exp(\ep)\alpha,\exp(-\ep)(1-\de-\alpha)\}$. Then a mechanism $M$ satisfies $(\ep,\de)$-DP if and only if it satisfies $f_{\ep,\de}$-DP \citep[Proposition 3]{dong2022gaussian}. The following proposition, based on \citet[Propositions 5 and 6]{dong2022gaussian}, gives a simple conversion between $f$-DP and $(\ep,\de)$-DP, by determining the linear functions which lower bound $f$. 

\begin{prop}\label{prop:ftoDP}
Let $f$ be a symmetric tradeoff function. If a mechanism satisfies $f$-DP, then it satisfies $(\ep,\de)$-DP provided that $(1-\de)-\exp(\ep)\alpha\leq f(\alpha)$ for all $\alpha\in[0,1]$. 
\end{prop}
\begin{proof}
We need to show that $f_{\ep,\de}(\alpha)\leq f(\alpha)$ for all $\alpha \in [0,1]$. By symmetry of $f$ and $f_{\ep,\de}$, the condition stated is sufficient.
\end{proof}
If the tradeoff function $f$ makes the inequalities of Definition \ref{def:fDP} tight, then by Proposition \ref{prop:ftoDP} the tightest $(\ep,\de)$-DP guarantee takes a tangent line of $f$ and sets $(1-\de)$ to be the $y$-intecept and $-\exp(\ep)$ to be its slope. This approach gives a precise conversion from $f$-DP to $(\ep,\de)$-DP, which we use in Theorem \ref{thm:privacyCost}. In fact, there is a stronger duality between $f$-DP and a family of $(\ep,\de(\ep))$-DP characterizations, described in \citet[Propositions 5 and 6]{dong2022gaussian}. Figure \ref{fig:f_to_DP} illustrates the conversion from $f$-DP to $(\ep,\de)$-DP. 

An important property of both $(\ep,\de)$-DP and $f$-DP is that it is robust to \emph{post-processing}. That is, if a mechanism satisfies DP, then applying any  deterministic or randomized algorithm to the output cannot degrade the DP guarantee. This property is related to data processing inequalities.

\begin{prop}[Post-processing: \citealp{dwork2014algorithmic,dong2022gaussian}]\label{prop:postprocess}
Let $M$ be a mechanism with output space $\mscr Y$, $f$ a tradeoff function, $\ep\geq 0$, and $\de\in [0,1]$. Let $\mathrm{Proc}$ be a potentially randomized mapping from $\mscr Y$ to $\mscr Z$. Then 
\begin{enumerate}
    \item if $M$ satisfies $(\ep,\de)$-DP, then $\mathrm{Proc}\circ M$ satisfies $(\ep,\de)$-DP;
    \item if $M$ satisfies $f$-DP, then $\mathrm{Proc}\circ M$ satisfies $f$-DP. 
\end{enumerate}
\end{prop}

\subsection{Exponential Mechanism}
Having established the definitions of both $(\ep,\de)$-DP and $f$-DP, there remains the question of how to construct a privacy mechanism for a given statistical task. A general and powerful technique, and one that will be the focus of this paper, is the \emph{exponential mechanism}~\citep{mcsherry2007mechanism}.
%
Given a utility function $\g_D$, where large values of $\g_D$ indicate higher utility,  the {exponential mechanism} samples from the unnormalized density $\twid \pi_D(x) = \exp(\g_D(x))\pi_0(x)$, where $\pi_0$ is a base measure. This mechanism satisfies $(2/\Delta,0)$-DP where  $\Delta$ is the \emph{sensitivity} of $\g_D$:
\[\Delta \geq \sup_{d(D,D')\leq 1}\sup_x |\g_D(x)-\g_{D'}(x)|.\]
 Often $\pi_0$ is chosen to be Lebesgue measure, but it can also be chosen to be a probability measure similar to a prior \citep{wang2015privacy,minami2016differential,dimitrakakis2017differential}. In infinite-dimensional function spaces, there is no translation-invariant measure, so a nontrivial base measure must be used \citep{awan2019benefits}.
Many statistical tasks can be expressed as finding the solution to a minimization or maximization problem of some objective function (e.g., log-likelihood function, sum of squared error, or a general empirical risk function). For these tasks, it is natural to choose the utility function in the exponential mechanism to be some transformation of such an objective function. For example, \citet{reimherr2019kng} show that when an objective function $\xi_D(x)$ is strongly convex, sampling from the exponential mechanism with utility function $g(x)=-\lVert \nabla \xi_D(x)\rVert$ results in an estimator which satisfies $x^*=\arg\min_x \xi_D(x)+O_p(n^{-1})$. 
Though the exponential mechanism was designed with $(\ep,0)$-DP in mind, it has been shown that when the utility function satisfies additional assumptions such as concavity, Lipschitz continuity, or strong concavity, the exponential mechanism may satisfy $(\ep,\de)$-DP \citep{minami2016differential,dimitrakakis2017differential} even when the sensitivity $\Delta$ is infinite.

While the exponential mechanism is very flexible and offers high utility guarantees, sampling $\exp(g_D(x))$ exactly is generally very challenging. 
While specific implementations of the exponential mechanism sometimes have efficient sampling schemes (e.g., \citealp{bassily2014private,bassily2014privateRevisited,asi2020instance,asi2020near}), in general, more sophisticated computational sampling techniques are needed. For example, \citet{chaudhuri2012near,chaudhuri2013near} and \citet{awan2019benefits} use a Gibbs sampler to implement the exponential mechanism in the application of principal component analysis, using heuristics to argue convergence. \citet{reimherr2019kng} use MCMC implementations of their proposed $K$-norm gradient (KNG) mechanism, but leave considerations of the cost of the implementation for future work. \citet{snoke2018pmse} propose an instance of the exponential mechanism for synthetic data, which they sample using the Metropolis algorithm, without considering the privacy cost of the sampler.

\subsection{Rejection sampling}
Given the structure of the unnormalized density, sampling from $\exp(g_D(x))$ is often well suited to rejection sampling. Given an unnormalized target density $\pi(x)\propto \twid \pi(x)=\exp(g(x))$, which is difficult to sample from, and a simpler {\em proposal} density $U(x)$ which satisfies $\twid \pi(x)\leq c U(x)$ for some $c$ and all $x$, a rejection sampler {\red draws} $X\sim U(x)$ and accepts the sample with probability $\twid\pi(X)/(cU(X))$. This process is repeated until a sample is accepted, and it is easy to show that the accepted sample is distributed as $X\sim \pi(x)$. 
The requirements to implement a rejection sampler are that we can evaluate $\twid\pi(x)$, and determine $U(x)$ and $c$ which satisfy the above inequality. We will call these samplers \emph{simple rejection samplers} when we need to distinguish these from adaptive rejection samplers, which we introduce later in this section. See \citet{martino2018independent} for an extensive introduction to rejection samplers.

The marginal probability of accepting a sample at any particular iteration from a simple rejection sampler is $p=c^{-1}\int \twid\pi(x) \ dx$, so that the number of iterations $T$ before an accepted sample follows a geometric distribution: $T\sim \mathrm{Geom}(p)$.
 In this paper we assume that the geometric distribution has support $1,2,3,\ldots$, so that this has pmf $P(T=k) = (1-p)^{k-1}(p)$ for $k= 1,2,3,\ldots$.

While rejection samplers allow exact samples to be drawn from an intractable target distribution, the acceptance probability $p$ typically decays exponentially with dimension, making them suitable only for low-dimensional problems.
Adaptive rejection samplers attempt to address this shortcoming, and proceed by producing a {\em sequence} of upper bounds $U_n(x)$ and constants $c_n$ such that $\twid \pi(x)\leq c_n U_n(x)$ and such that the acceptance probability increases with $n$. Just like a simple rejection sampler, conditional on acceptance, adaptive rejection samplers produce samples $X\sim \pi(x)$.  Typically, the upper bounds are updated stochastically, using the information from the previously rejected samples. 
While this 
minimizes the number of evaluations of $\pi$, 
the acceptance probabilities update in a manner depending on the target $\pi$, making the runtime difficult to analyze. Alternatively, the upper bound can be updated in a deterministic manner such as in \citet{leydold2002variants}, which makes understanding the runtime much simpler. While deterministic updates require more evaluations of $\pi$, they can potentially result in upper bounds that converge to $\pi$ much faster resulting in a tradeoff. 

 With adaptive rejection sampling, the marginal probability of accepting a sample at iteration $n$ is $p_n=\frac 1{c_n}\int \pi(x) \ dx$. However, as the acceptance probability changes over time, the runtime $T$  to accept one sample is no longer geometric, but has pmf $P(T=k)=p_t\prod_{i=1}^{k-1} (1-p_i)$, for $k=1,2,3,\ldots$.

 \section{Privacy risk of rejection sampling}
 In this section we characterize the privacy cost of a rejection sampler, {\red when we allow the adversary to have access to both the accepted sample as well as the runtime.} Recall that if a rejection sampler is run until acceptance, then the accepted sample is an exact sample from the target distribution. 
 Thus, the only increased privacy risk from using this algorithm is due to the runtime. 
 We will measure the privacy risk of this side-channel in terms of $\ep$-DP, $(\ep,\de)$-DP, and $f$-DP. {\red We show that for the exponential mechanism, the privacy cost of a rejection sampler's runtime is non-negligible.}
 
 \begin{assumption}
For a rejection sampler, we assume that along with the published accepted sample, the runtime is also available to an attacker. We assume that for all databases $D$ and for all $x$ in the domain, the evaluations $\g_D(x)$ take the same time to evaluate. As such, the runtime is proportional to the number of iterations in the sampler. {\red Thus for the rest of the paper, the runtime will simply refer to the number of iterations in the sampler.}

{\red Note that while the proposal distribution $U_D(x)$, target $\exp(g_D(x))$, and threshold $c_D$ may all depend on $D$, none are directly available to the attacker.}
 \end{assumption}
 
 \begin{rem}
 Many utility functions used in the exponential mechanism can be expressed as empirical risks \citep{bassily2014private,bassily2014privateRevisited,reimherr2019kng,wang2019differentially}.  
 In this case, {\red assuming that the database size $n$ is fixed,} ensuring that the time to evaluate $\g_D(x)$ is constant is equivalent to ensuring that the contributions to the empirical risk from each individual take constant time. This is 
 in line with the techniques used in \citet{haeberlen2011differential} who split each query into sub-queries which are evaluated on each member of the dataset. 
 \end{rem}

First we will study the privacy cost of the rejection sampling runtime in terms of $\ep$-DP. Proposition \ref{prop:RSnoDP} states that the runtime of a rejection sampler violates $\ep$-DP unless the probability of acceptance is constant across databases. To prove this, recall that $\ep$-DP is measured by the max-divergence.  Lemma \ref{lem:geomRatio} shows that the symmetric max-divergence between two geometric random variables is unbounded whenever the parameters differ, and the proposition follows easily from this.

 \begin{lem}\label{lem:geomRatio}
 Let $p,q\in (0,1)$ and let $X \sim \mathrm{Geom}(p)$ and $Y \sim \mathrm{Geom}(q)$. Then
 \[D_{\infty}(X||Y)=
 \begin{cases}
 \log(p/q)&\text{ if } p\geq q\\
 \infty& \text{ if }p<q.
 \end{cases}\]
 {\red Thus, $D_\infty^S(X,Y)=\infty$ whenever $p\neq q$.}
 \end{lem}
 \begin{proof}
 As all geometric random variables, with parameter in $(0,1)$, are equivalent measures on the positive integers, it suffices to determine an upper bound on $\log \frac{P(X=k)}{P(Y=k)}$ for $k \in \{1,2,\ldots\}$. This quantity can be expressed as
 \begin{align*}
     \log \frac{P(X=k)}{P(Y=k)}=\log \frac{(1-p)^{k-1}p}{(1-q)^{k-1}q}&= \log\left( \frac{p(1-q)}{q(1-p)}\right) + k\log\left( \frac{1-p}{1-q}\right).
 \end{align*}
 We see that this quantity is linear in $k$. The slope is non-positive if and only if $p\geq q$, in which case the maximum value is achieved at $k=1$, giving the value $\log(p/q)$. When $p<q$, the slope is positive, and as $k\rightarrow \infty$, the quantity is unbounded.
 \end{proof}
 
 \begin{prop}\label{prop:RSnoDP}
 Let $\{M_D\mid D\in \mscr D\}$ be a privacy mechanism, let $p_D$ be the probability of acceptance for a rejection sampler run on $M_D$, call  $T_D$ the runtime of the rejection sampler {\red which is distributed $\mathrm{Geom}(p_D)$}, and call $X$ the accepted sample. If there exists $D,D'\in \mscr D$ such that $d(D,D')\leq 1$ and $p_D\neq p_{D'}$, then the mechanism that releases $(X,T)$ does not satisfy $\ep$-DP for any $\ep>0$. 
 \end{prop}
 \begin{proof}
 By post-processing (Proposition \ref{prop:postprocess}), we get a lower bound on the privacy cost by only considering $T$. If there exists $D$ and $D'$ such that $d(D,D')\leq 1$ and $p_D\leq p_{D'}$, then by Lemma \ref{lem:geomRatio} the symmetric max-divergence is unbounded, and the result follows. 
 \end{proof}

Theorem \ref{thm:privacyCost} gives a more precise characterization of the privacy loss due to rejection sampling as measured by $f$-DP and $(\ep,\de)$-DP. 
For the former, we bound the tradeoff function of the geometric variables with the tradeoff function for exponential variables, which allows for a simpler formula. This bound is tighter for small acceptance probabilities. We use the likelihood ratio property of the exponential distribution along with some properties of convex functions to get the formula in equation~\eqref{eq:exp_bnd}. We then use Proposition \ref{prop:postprocess} to convert the $f$-DP guarantee to $(\ep,\de)$-DP guarantees.
 
 \begin{figure}
    \centering
    \begin{subfigure}[t]{.48\textwidth}
    \includegraphics[width=\linewidth]{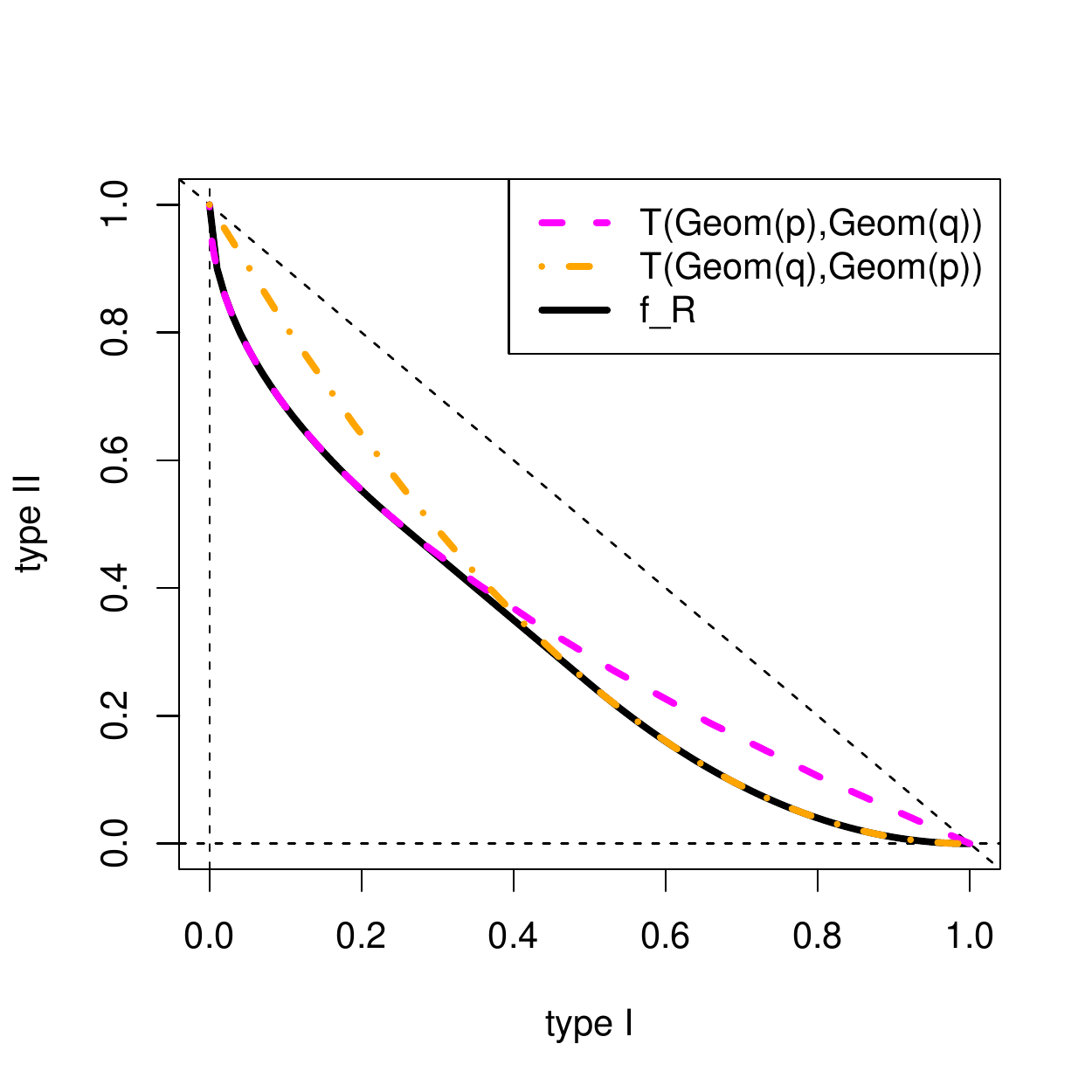}
    \end{subfigure}
		\begin{subfigure}[t]{0.48\textwidth}
    \includegraphics[width=\linewidth]{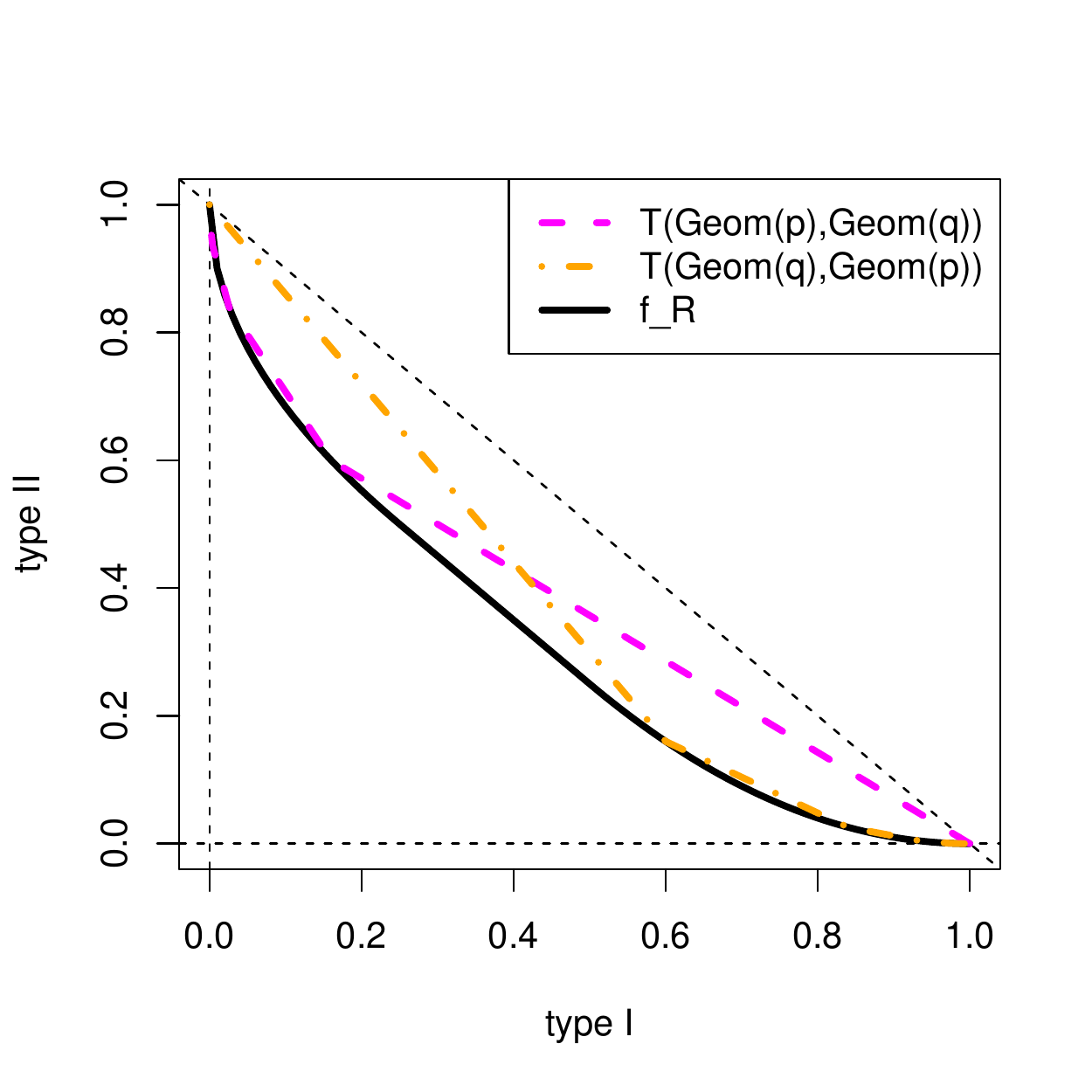}
\end{subfigure}
    \caption{The tradeoff functions of $T(\mathrm{Geom}(p),\mathrm{Geom}(q))$ and $T(\mathrm{Geom}(q),\mathrm{Geom}(p))$, along with $f_R$ from Theorem \ref{thm:privacyCost}. We fix $R=2$. In the left plot $q=.1$ and in the right $q=.6$. $p=1-(1-q)^R>q$. We see that the approximation $f_R$ is more accurate for smaller $q$.}
    \label{fig:fR}
\end{figure}

 \begin{thm}\label{thm:privacyCost}
 Let $(\mscr D,d)$ be a metric space of databases, and let $T_D$ be the runtime of a rejection sampler for database $D$ which has acceptance probability $p_D$. Note that $T_D\sim \mathrm{Geom}(p_D)$. Call $R = \sup_{d(D,D')\leq 1} \frac{\log(1-p_D)}{\log(1-p_{D'})}$. The mechanism that releases the runtime $T_D$ 
 \begin{enumerate}
     \item satisfies $f_R$-DP, where 
     \begin{align}f_R(\alpha)= \begin{cases}
     1-\alpha^{1/R}& \alpha \leq R^{R/(1-R)}\\
     -\alpha + R^{R/(1-R)} + 1-R^{1/(1-R)}&R^{R/(1-R)}< \alpha< 1- R^{1/(1-R)}\\
     (1-\alpha)^R&\alpha\geq1-R^{1/(1-R)},
     \end{cases}\label{eq:exp_bnd}\end{align}
     \item satisfies $(\ep,\de(\ep))$-DP for all $\ep\geq 0$, where $\de(\ep)=(1-1/R)\exp\left(\frac{-\ep-\log(R)}{R-1}\right)$,
     \item satisfies $(\ep(\de),\de)$-DP for all $0<\de\leq (R-1)R^{R/(1-R)}$, where $\ep(\de) = \log(1/R) + (R-1)(\log(1/\de)+\log(1-1/R))$.
 \end{enumerate}
 \end{thm}
 \begin{proof}
 We begin by establishing the form of $f_R$, and then use Proposition \ref{prop:ftoDP} to produce $(\ep,\de)$-DP guarantees. We first show that by bounding the tradeoff function of the exponential distribution, we get bounds for geometric variables as well.
 Call $\lambda_D=-\log(1-p_D)$. Recall that if $X_D\sim \mathrm{Exp}(\lambda_D)$, then $\lfloor X_D\rfloor+1\sim \mathrm{Geom}(p_D)$. By Proposition \ref{prop:postprocess}, we have that $T(\mathrm{Exp}(\lambda_{D}),\mathrm{Exp}(\lambda_{D'}))=T(X_D,X_{D'})\leq T(\lfloor X_D\rfloor+1,\lfloor X_{D'}\rfloor +1)=T(\mathrm{Geom}(p_D),\mathrm{Geom}(p_{D'}))$, where $T(\cdot,\cdot)$ represents the tradeoff function.

 Next, we will derive the tradeoff function $T(\mathrm{Exp}(\lambda_{D}),\mathrm{Exp}(\lambda_{D'}))$ assuming that $\lambda_D>\lambda_{D'}$. 
 Let $p_{\lambda_D}(x)$ be the pdf of $\mathrm{Exp}(\lambda_D)$. Note that $(p_{\lambda_{D'}}(x)/p_{\lambda_D}(x)) = (\lambda_{D'}/\lambda_D)\exp(x(\lambda_D-\lambda_{D'}))$ is an increasing function of $x$. By the Neyman Pearson Lemma, the most powerful test has a rejection region of the form $x\geq T$. The type I error is $\alpha = \exp(-\lambda_D T)$ and type II is $\beta = 1-\exp(-\lambda_{D'}T)$. Expressing $\beta$ as a function of $\alpha$ gives $\beta = 1-\alpha^{\lambda_{D'}/\lambda_D}\geq 1-\alpha^{1/R}$. Thus, we have that $T(X_D,X_{D'})\geq 1-\alpha^{1/R}$. We also need a lower bound on $T(X_{D'},X_{D})$. Note that $T(X_{D'},X_D)$ is the inverse of $T(X_D,X_{D'})$. By taking the inverse of $1-\alpha^{1/R}$, we have $T(X_{D'},X_{D})\geq (1-\alpha)^R$. 
 
 To get a single bound on both $T(X_D,X_{D'})$ and $T(X_{D'},X_{D})$, we take the convex hull of $\min\{1-\alpha^{1/R},(1-\alpha)^R\}$, which we claim has the form $f_R(\alpha)$ as stated in 1. To this end, we first verify that 
 \begin{equation}\label{eq:tradeExp}
 1-\alpha^{1/R}\leq 1-R\alpha\leq (1-\alpha)^R,
 \end{equation}
 for all $0\leq \alpha\leq R^{R/(1-R)}${\red, so that over this range of $\alpha$, the convex hull just equals $1-\alpha^{1/R}$ as required by the first line of Equation~\eqref{eq:exp_bnd}.  To establish the first inequality of Equation \eqref{eq:tradeExp}, note that $f(\alpha)=1-\alpha^{1/R}$ is a convex function; this can be seen either by diifferentiating it twice, or from the fact that it is a tradeoff function. The straight line $1-R\alpha$ intersects this curve at $\alpha=0$ and $\alpha = R^{R/1-R}$, and for intermediate values of $\alpha$, forms a chord segment.
 From convexity, this chord lies above the curve.} 
 For the second inequality of Equation \eqref{eq:tradeExp}, observe that $(1-\alpha)^R$ is also convex. 
 We can easily verify that the line $1-R\alpha$ is the tangent at $\alpha=0$. The second inequality then follows from the fact that a convex function is lower bounded by its tangent.
 This justifies the first line of $f_R(\alpha)$ in Equation~\eqref{eq:exp_bnd}.
 By symmetry, we also have that the third line is correct. 
 
 {\red For the middle inequality, we note that the curves $1-\alpha^{1/R}$ and $(1-\alpha)^{R}$ have slope $-1$ at the points $R^{R/(1-R)}$ and  $1-R^{1/(1-R)}$ respectively.
 It is easily verified that the straight line $g(\alpha) = -\alpha + R^{R/(1-R)} + 1-R^{1/(1-R)}$ intersects the two curves at these two points, and has slope $-1$.
 It is thus tangent to both curves, and from convexity, lies below both of them. Altogether, we conclude that  $f_R(\alpha)$ is the appropriate convex hull. }

 To get the formulas in 2. and 3., recall that the mechanism satisfies $(\ep,\de)$-DP if the line {\red$(1-\de)-\exp(\ep) \alpha$} is a lower bound for the tradeoff function $f_R(\alpha)$. To get the tightest $(\ep,\de)$-DP guarantees, we characterize the supporting linear functions. By symmetry, it suffices to determine the tangent lines of $1-\alpha^{1/R}$ for values $0\leq \alpha\leq R^{R/(1-R)}$. We calculate the derivative as $\frac{d}{d \alpha} (1-\alpha^{1/R})=\frac{-1}{R} \alpha^{1/R-1}$. Set $-\exp(\ep) = \frac{-1}{R} \alpha^{1/R-1}$, which has the solution $\ep = \log(1/R) + (1/R-1)\log \alpha$. 
 
 The line with slope $-\exp(\ep)=\frac{-1}{R} \alpha^{1/R-1}$ that passes through $(\alpha,1-\alpha^{1/R})$ is $y-(1-\alpha^{1/R})=\frac{-1}{R} \alpha^{1/R-1}(x-\alpha)$, which has $y$-intercept $1-\de=1-\alpha^{1/R}(1-1/R)$, giving $\de = \alpha^{1/R}(1-1/R)$.  Eliminating $\alpha$ from the equations $\ep = \log(1/R) + (1/R-1)\log \alpha$ and $\de = \alpha^{1/R}(1-1/R)$ gives the expressions in parts 2. and 3. in the theorem statement.  Note that $\de(0)=(R-1)R^{R/(1-R)}$, so for any $\de>\de(0)$, the mechanism satisfies $(0,\de)$-DP, but this is a strictly weaker guarantee than $(0,\de(0))$-DP. 
 \end{proof}
 
  The approximation in Theorem \ref{thm:privacyCost} improves for smaller probabilities of acceptance, as seen in Figure \ref{fig:fR}. Intuitively, this is because the approximation of a geometric variable as an exponential is more accurate for smaller probabilities of acceptance. As rejection samplers typically have small rejection probabilities, the privacy guarantees of Theorem \ref{thm:privacyCost} are quite accurate for rejection samplers of interest. {\red  In Table \ref{tab:Rconvert}, we give a few examples converting the quanitity $R$ to $(\ep,\de)$-DP guarantees. We see that even with a small $R$ of $1.1$, there is a nontrivial privacy cost. In Example \ref{ex:expMech}, we explore what values of $R$ we may expect in practice.}

 \begin{cor}\label{cor:privacyCost}
 Let $(\mscr D,d)$ be a metric space of databases, $M_D$ a privacy mechansim which satisfies $f$-DP and $T_D$ the runtime of a rejection sampler for $M_D$ which has acceptance probability $p_D$. Call $R = \sup_{d(D,D')\leq 1} \frac{\log (1-p_D)}{\log(1-p_{D'})}$. Then the privacy cost of $M_D$ along with the runtime is $f_R\otimes f$, where $f_R$ is defined in Theorem \ref{thm:privacyCost} and $\otimes$ is the \emph{tensor product} of two tradeoff functions \citep[Definition 5]{dong2022gaussian}.
 \end{cor}
 
    \begin{table}
  {\red
     \[\begin{array}{c|cccccc}
         \delta=& .1&.01&.001&10^{-4}&10^{-5}&10^{-6}  \\\hline
          R=2&.916&3.22&5.52&7.82&10.13&12.43\\
          R=1.1&0&.125&.356&.59&.82&1.05
     \end{array}\]}
     \caption{{\red The $(\ep(\de),\de)$-DP guarantee for a simple rejection sampler, where $R$ (defined in Theorem \ref{thm:privacyCost}) is either $2$ or $1.1$. The values $\ep(\de)$ appear in the table for each combination of $\delta$ and $R$.}}
     \label{tab:Rconvert}
 \end{table}

 In Corollary \ref{cor:privacyCost}, the tensor product $f\otimes g$, where $f = T(P,Q)$ and $g=T(P',Q')$ is defined as $f\otimes g=T(P\times P',Q\times Q')$, where $P\times P'$ is the product distribution  \citep[Definition 5]{dong2022gaussian}. In general, it is challenging to derive a closed form of $f\otimes g$. 
 
 \begin{rem}[Rejection Sampling Trivial for Location-Scale]
 For some distributions, it is easy to build a rejection sampler, with constant acceptance probability. For example, suppose that the mechanism $\{M_D\mid D\in \mscr D\}$ is location-scale (e.g., $K$-norm mechanisms: \citealp{hardt2010geometry,awan2020structure}).  In this case, we build a rejection sampler for a default distribution in the family, and transform after sampling. Then we have a rejection sampler where the acceptance rate is independent of the dataset. 
 \end{rem}

 {\red
 While Theorem \ref{thm:privacyCost} describes the privacy cost of a rejection sampler's runtime, it is phrased in terms of the quantity $R$, which may be unintuitive. In the following example, we show that for a generic exponential mechanism, with an arbitrary set of proposal distributions, $R$ is lower bounded by  $\exp(\ep)$, and may even be infinite.
 
 \begin{ex}[Exponential Mechanism]\label{ex:expMech}
 Recall that the exponential mechanism results in a target distribution of the form $\twid \pi_D = \exp(g_D(x))$, which usually satisfies $\exp(-\ep/2)\leq \frac{\twid \pi_{D'}(x)}{\twid \pi_D(x)}\leq \exp(\ep/2)$ for adjacent databases $D$ and $D'$ (the integrating constants may also differ by a factor of at most $\exp(\pm\ep)$) \citep{mcsherry2007mechanism}. Let $\mscr U$ be a family of proposal distributions, and for each database $D$, let $c_D$ and $U_D$ be the optimal proposal distribution from $\mscr U$ such that $\twid \pi_D\leq c_D U_D(x)$, where by optimal, we mean that the acceptance probability $p_D = \frac{\int \twid \pi_D(x) \ dx}{c_D}$ is maximized; or equivalently $c_D$ is minimized. Then, from the following inequality,
 \[\twid \pi_{D'}(x)\leq \exp(\ep/2) \twid \pi_D(x)\leq \exp(\ep/2) c_D U_D(x),\]
 we see that $\exp(\ep/2)c_D$ and $U_D$ offer a (potentially inferior) proposal distribution for $\twid \pi_{D'}$. Using this relationship between the proposal distributions of $D$ and $D'$, we can give a bound for the acceptance probability for $D'$ based on the acceptance probability for $D$:
 \[p_{D'} = \frac{\int \twid \pi_{D'}(x) \ dx}{c_{D'}}\geq \frac{\int \twid \pi_{D'} \ dx}{\exp(\ep/2) c_D}\geq \frac{\exp(-\ep/2) \int \twid \pi_{D} \ dx}{\exp(\ep/2) c_D} = \exp(-\ep) p_D.\]
 Call $p^*$ the highest acceptance probability over all possible databases $D$. Then the quantity $R$ that appears in Theorem \ref{thm:privacyCost} can be expressed as 
 \begin{equation}\label{eq:expMechR}
     R=\frac{\log (1-p^*)}{\log(1-\exp(-\ep) p^*)}.
     \end{equation}
     Note that as $p^*\rightarrow 1$ in Equation \eqref{eq:expMechR}, $R$ diverges to infinity. We can also get a lower bound on $R$:
     \begin{equation}
     R=\frac{\log (1-p^*)}{\log(1-\exp(-\ep) p^*)}
     \geq \lim_{p\rightarrow 0} \frac{\log(1-p)}{\log(1-\exp(-\ep)p)}
     \overset{L'H}{=} \lim_{p\rightarrow 0} \frac{1-\exp(-\ep)p}{\exp(-\ep)(1-p)}
     =\exp(\ep),
 \end{equation}
 where $\overset {L'H}{=}$ indicates the use of L'H\^opital's rule, and we used the fact that $\log(1-p)/\log(1-\exp(-\ep)p)$ is increasing in $p$ for all $p\in (0,1)$ and $\ep>0$; to see this, we compute the derivative with respect to $p$:
 \begin{equation}
     \frac{(1-p)\log(1-p)-(\exp(\ep)-p)\log(1-p \exp(-\ep))}{(\exp(\ep-p)(1-p)(\log(1-p\exp(-\ep)))^2}.\label{eq:deriv}
 \end{equation}
 We see in \eqref{eq:deriv} that the denominator is positive so long as $0<p<1$. The numerator of \eqref{eq:deriv} can be expressed as 
 \[\sum_{n=2}^\infty p^n\left (\frac{1}{n(n-1)}\right)\left(1-\exp(-\ep(n-1))\right),\]
 which we can see is positive and finite for all $\ep>0$ and $p\in (0,1)$.

 \end{ex}
 }

 {\red \begin{rem}[Parallelization/Batching]
 Suppose that we have a simple rejection sampler targeting $\twid \pi_D$ with acceptance probability $p_D$. We could consider a parallelized implementation as follows: run the sampler on $k$ nodes; when the first sample is accepted, return the sample and the runtime and abort the other instances of the sampler. In this scheme, the runtime is distributed as $\min\{G_1,\ldots, G_k\}\sim \mathrm{Geom}(1-(1-p_D)^k)$, where $G_i$ are independent $\mathrm{Geom}(p_D)$ random variables. Now, suppose for two adjacent databases $D$ and $D'$ that 
 \[\frac{\log(1-p_D)}{\log(1-p_{D'})}=R,\]
 which is the quantity in Theorem \ref{thm:privacyCost} that governs the privacy cost of the runtime.  Then, in the parallelized scheme, we have 
 \[\frac{\log (1-[1-(1-p_D)^k])}{\log(1-[1-(1-p_{D'})^k])} 
 =\frac{\log((1-p_D)^k)}{\log((1-p_{D'})^k)}
 =\frac{k\log(1-p_D)}{k\log(1-p_{D'})}
 =R.\]
 We see that parallelization does not affect the privacy cost of the runtime.
 
 Similarly, one could decide to run the rejection sampler for a fixed number of iterations, say $m$ before checking if one of the samples is accepted, and then repeating if necessary. This may be useful in combination with parallelization, since communication between the nodes could be a bottleneck. With batching, the runtime until a sample is accepted is $\mathrm{Geom}(1-(1-p_D)^m)$, which is the same runtime as in the parallelization. By the same reasoning, batching also does not affect the privacy cost of the runtime.
 
 In Section \ref{s:results}, we develop samplers with data-independent runtime. As such, parallelizing or batching the samplers in the manner described above maintains the property that the runtime is data-independent, while potentially giving a significant speed up.
 \end{rem}
 }

 \subsection{Privacy risk of adaptive rejection sampling}
 In this section, we analyze the privacy risk of adaptive rejection samplers. Often adaptive rejection samplers update the proposal in a stochastic manner, based on the target value at previously rejected samples. In this section, we consider the setting where the proposal is updated in a deterministic manner, such as in \citet{leydold2002variants}. 
 We show that unless the acceptance probabilities converge in a strong sense, an adaptive rejection sampler will not satisfy $\ep$-DP for any finite $\ep$. 

 \begin{prop}\label{prop:adaptiveCost}
 Let $\mscr D$ be the space of databases and $\{M_D\mid D\in \mscr D\}$ a privacy mechanism which satisfies $\ep$-DP. Let $(p_i^D)_{i=1}^\infty$ be the sequence of acceptance probabilities for an adaptive rejection sampler for $M_D$. Call $T_D$ the runtime of the adaptive sampler for $M_D$, which has pmf $P(T_D=t)=p_t^D\prod_{i=1}^{\red t-1}(1-p_i^D)$. Then releasing a sample from $M_D$ as well as the runtime $T_D$ satisfies $(\ep+\ep_T)$-DP, where 
 \[\ep_T\geq \log(p_t^D/p_{t}^{D'}) + \sum_{i=1}^{t-1}\log\left(\frac{1-p_i^D}{1-p_i^{D'}}\right),
 \]
 for all $t\geq 1$ and all $d(D,D')\leq 1$. {\red If there exists a constant $c$ such that $p^D_1\geq c>0$} for all $D$, then the value $\ep_T$ is finite if and only if the sequence {\red $\left(\sum_{i=1}^t \log\frac{1-p_i^D}{1-p_i^{D'}}\right)_{t=1}^\infty$} is universally bounded for all $d(D,D')\leq 1$. 
 \end{prop}
 \begin{proof}
 For readability, we set $p_i\defeq p_i^D$ and $q_i\defeq p_i^{D'}$.  We require that $\log \frac{P(T_D=t)}{P(T_{D'}=t)}\leq \ep_T$ for all $d(D,D')\leq 1$. A little algebra gives the expression for $\ep_T$. 
 
Next, $\ep_T$ is finite if and only if $\log \frac{P(T_D=t)}{P(T_{D'}=t)}$ is bounded above and below for all $d(D,D')\leq 1$. Equivalently, this requires $\log \frac{p_t\prod_{i=1}^{t-1}(1-p_i)}{q_t\prod_{i=1}^{t-1}(1-q_i)}=
 \log \left(\frac{p_t}{q_t}\prod_{i=1}^{t-1} \left(\frac{1-p_i}{1-q_i}\right)\right)$ be universally bounded above and below for all $t$. Since  $\frac{p_t}{q_t}$ is bounded below by {\red $c$} and above by {\red $1/c$}, the previous quantity is bounded if and only if $\log \left(\prod_{i=1}^{t-1} \left(\frac{1-p_i}{1-q_i}\right)\right)=\sum_{i=1}^{t-1}\log\left(\frac{1-p_i}{1-q_i}\right)$ is bounded for all $t$. Relabelling $t-1$ to $t$ gives the final result. 
 \end{proof}
 
 Proposition \ref{prop:adaptiveCost} shows that unless the acceptance probabilities are very closely related, it is not guaranteed that an adaptive rejection sampler will satisfy $\ep$-DP for any finite $\ep$. In the following example, we explore a few special cases to highlight when we can or cannot expect the condition in Proposition \ref{prop:adaptiveCost} to hold. 
 \begin{ex}
 
  \begin{itemize}
     \item If there exists $i$ such that $p_i=1$ whereas $q_i<1$ or vice versa, then $\ep_T=\infty$.
     \item Suppose that $(1-q_i)=\alpha(1-p_i)$ where $\alpha\in (0,1)$. Then the above series is $\sum_{i=1}^t \log \frac{1-p_i}{1-q_i}=\sum_{i=1}^t \log \alpha\rightarrow \infty$. 
     \item To see that it is not sufficient for $\lim_{i\rightarrow \infty} \frac{1-p_i}{1-q_i}=1$, consider the following: let $(1-p_i)$ be any decreasing sequence with values in $(0,1)$. Set $(1-q_i) = \exp(-1/i) (1-p_i)$. Then $\log\left(\frac{1-p_i}{1-q_i}\right)= 1/i$ and so $\frac{1-p_i}{1-q_i}\rightarrow 1$. However, the sequence of partial sums $\sum_{i=1}^t  \log\left(\frac{1-p_i}{1-q_i}\right) = \sum_{i=1}^t 1/i$ diverges, and so the max-divergence is infinite. 
 \end{itemize}
 \end{ex}
 
 \begin{rem}
 In Proposition \ref{prop:adaptiveCost}, convergence of the series $\sum_{i=1}^\infty \log\frac{1-p_i}{1-q_i}$ is sufficient but not necessary. It is possible that the sequence of partial sums is bounded but does not converge. 
 \end{rem}
 
 Note that for most adaptive rejection samplers, it is difficult to derive expressions for $p_i$, so it may not even be possible to verify whether the condition in Proposition \ref{prop:adaptiveCost} holds or not. The takeaway is that in general, an adaptive rejection sampler is not guaranteed to preserve privacy unless it is carefully designed to do so. 
 
 \section{Rejection samplers with data-independent runtime}\label{s:results}
 
The previous section showed that a rejection sampler (either simple or adaptive) can result in an arbitrary amount of privacy loss through the runtime. 
The most direct way to avoid this is to ensure that the runtime does not depend on the dataset. \citet{haeberlen2011differential} propose making the runtime a constant, though this is not strictly necessary. Rather, when the runtime is a random variable (as with rejection sampling), we simply need that its distribution does not depend on the dataset.
 
 In this section we propose three modifications of the rejection sampling algorithm to ensure data-independent runtime. The first method, which requires the weakest assumptions, fixes the number of iterations independent of the dataset, based on a worst-case acceptance probability. This method has a constant runtime, but there is a small probability that a sample is not accepted, and we quantify the additional privacy cost. 
 The second method is based on the memoryless property of the geometric distribution, and introduces an additive random wait-time. 
 This approach however requires the integrating constant of the target distribution corresponding to the current database, as well as the acceptance probability of a worst-case database, which is often not realistic. 
 The third method avoids this by using instead upper and lower bounds for the target densities of all databases, chosen so that the ratio of the area for the upper and lower bounds is constant across databases. 
 Finally, we propose an adaptive rejection sampler with data-independent runtime, which is a modification of the (nearly) minimax optimal sampler of \citet{achddou2019minimax}. Our sampler is entirely automated, and only requires that the family of target densities is log-H\"{o}lder with fixed and known parameters
 
 We show in Section \ref{s:examples} that many commonly studied privacy mechanisms satisfy the assumptions of our methods allowing for our privacy-preserving rejection samplers to be applied. 
 
 \subsection{Constant runtime, truncated rejection sampling}\label{s:truncated}
 One clear way to remove the privacy leak due to the runtime is to choose a number of iterations independent of the database, based on a worst-case estimate of the acceptance probability across all databases. We then run the sampler for that many iterations, and publish one of the accepted samples. In this case, the runtime is fixed, and does not leak any privacy. However, it is not guaranteed that an accepted sample is found within the pre-determined number of iterations, and the probability of this event \emph{does} depend on the database. This probability can be reduced by increasing the number of iterations, but  this also increases the runtime of the algorithm.
 
 Of the methods we propose, the algorithm in Proposition \ref{prop:truncated} requires the weakest assumptions in that the only knowledge we require is a lower bound on the acceptance probability across the databases. However, there is a small probability that no samples are accepted in the prescribed number of iterations, which negatively impacts both the privacy and the utility of the mechanism. Proposition \ref{prop:truncated} characterizes the increased cost to privacy of the truncated sampler in terms of $(\ep,\de)$-DP. 
 
 \begin{prop}\label{prop:truncated}
Let $\{M_D\mid D\in \mscr D\}$ be a family of mechanisms satisfying $(\ep_0,\de_0)$-DP and $(U_D,c_D)$ be such that $\twid\pi_D\leq c_DU_D$ where $\twid \pi_D$ is an unnormalized density for $M_D$. Assume that $\alpha_0\leq 1/c_D\int \twid\pi_D(x) \ dx$  for all $D$, that is, $\alpha_0$ is a lower bound on the acceptance probability in the rejection sampler across all databases. Given $\delta>0$, run the sampler for $N = \frac{\log(1/\de)}{\log(1/(1-\alpha_0))}$ iterations. If there is an accepted proposal, publish the first one; if not, publish an arbitrary output (such as one more draw from the proposal). Releasing the output as well as the runtime of this algorithm satisfies $(\ep_0, \delta_0+\delta)$-DP. 
\end{prop}
\begin{proof}
First note that the runtime is constant for all $D$, so there is no privacy leak there. Next, note that conditional on the event that an accepted proposal is found, there is no additional privacy leak. So, we need to determine the probability that an accepted proposal is not found:
\begin{align*}
    P(\text{none accepted}) &= (1-P(\text{accept}))^N\leq (1-\alpha_0)^N=(1-\alpha_0)^{\frac{\log(\de)}{\log(1-\alpha_0)}}=\delta.
\end{align*}
By itself, simply publishing whether a sample is accepted or not satisfies $(0,\de)$-DP. By post-processing (Proposition \ref{prop:postprocess}), we can upper bound the privacy cost by instead considering if we observe both an output from $M_D$ as well as whether the algorithm has accepted or rejected a sample. This is a composition of an $(\ep_0,\de_0)$-DP mechanism with a $(0,\de)$-DP mechanism. By composition \citep[Theorem 3.16]{dwork2014algorithmic} the result satisfies $(\ep_0,\de_0+\de)$-DP.
\end{proof}

A benefit of the algorithm in Proposition  \ref{prop:truncated} is that it can be vectorized and is embarrassingly parallelizable. Another benefit is that $N$ grows only in the log of $1/\delta$. By increasing the number of iterations $N$, the increased $\de$ can be reduced exponentially.  The two major downsides are that the algorithm must be run much longer than a simple rejection sampler, and that it is not guaranteed that an accepted sample is found, which reduces both the privacy and utility. If no samples are accepted, then the output does not follow the correct distribution, introducing error in the sampling approximation. We see that we are able to remove the runtime side-channel, but at the cost of a small ``delta'' and loss of utility. In the next two subsections, we show that with slightly stronger assumptions, we are able to obtain both perfect sampling as well as data-independent runtime. 

\subsection{Additive geometric wait-time}\label{s:additive}
In this section, we use the memoryless property of the geometric distribution to introduce an additive wait time based on a lower bound on the acceptance probability. The result is that the runtime of the algorithm is geometric with acceptance rate equal to the worst-case dataset (or a lower bound on the acceptance probability).

The benefit of this method over the truncated rejection sampler is that a sample from the correct distribution is guaranteed, and the runtime is still independent of the database. The downside is that the acceptance probability (or equivalently the integrating constant) for the present database is required as well as a bound on the worst-case acceptance probability. Typically, rejection samplers do not assume that the integrating constant is known, however for low dimensional problems (e.g., $\leq 3$), it may be possible to numerically evaluate the integral. 

Lemma \ref{lem:memoryless} illustrates the memoryless property of the geometric distribution. Given a simple rejection sampler with acceptance probability $q$, we can add a random wait time to result in a total runtime that is distributed as $\mathrm{Geom}(p)$ for $p\leq q$. So, across databases, we can make all of the runtimes equal in distribution, calibrated to a worst case acceptance probability.

\begin{lem} \label{lem:memoryless}
Let $0 < p \leq q< 1$. Given $X_2\sim \mathrm{Geom}(q)$, set $X_1=X_2$ with probability $p/q$ and otherwise $X_1=X_2+\Delta$, where $\Delta\sim \mathrm{Geom}(p)$. Then $X_1\sim \mathrm{Geom}(p)$. 
\end{lem}
\begin{proof}
Let $t\in \{1,2,\ldots, \infty\}$. Then 
\begin{align*}
    P(X_1=t)&=\frac{p}{q} P(X_2=t) + (1-p/q)P(X_2+\Delta=t)\\
    &=\frac{p}{q}(1-q)^{t-1}q + \frac{(q-p)}{q} \sum_{x=1}^{t-1}P(X_2=x)P(\Delta=t-x)\\
    &=p(1-q)^{t-1} + \frac{q-p}{q} \sum_{x=1}^{t-1} (1-q)^{x-1} q(1-p)^{t-x-1}p\\
    &=p(1-p)^{t-1},
\end{align*}
which is the pmf of $\mathrm{Geom}(p)$, as desired. To achieve the last line in the equations, we used the partial sum formula for a geometric series, and simplified the result.
\end{proof}

\begin{thm}\label{thm:memoryless}
Let $D\in \mscr D$ be a database and $\{\pi_D\mid D\}$ be the normalized target densities. Assume that for each $\pi_D$, we have normalized densities $U_D(x)$ as well as constants $c_{D}$ such that for all $x$, $\pi_D(x) \leq c_{D} U_D(x)$. Suppose we know a constant $c$ satisfying $c \geq \sup_D c_D$. Consider the following scheme:
\begin{enumerate}
    \item Run a rejection sampler, proposing from $U_D(x)$ and targeting $\pi_D(x)$ until acceptance
    \item Call the accepted sample $X$. Also draw $Y \sim \mathrm{Unif}(0,1)$. 
    \item If $ Y < c_D/c$, publish $X$, else wait for $\mathrm{Geom}(1/c)$ cycles before publishing $X$. 
\end{enumerate}
Then $X\sim \pi_D$, and the wait time follows $\mathrm{Geom}(1/c)$, which does not depend on $D$. 
\end{thm}

As compared to the truncated rejection sampler of Section \ref{s:truncated}, Theorem \ref{thm:memoryless} offers a perfect sampler with data independent runtime. This is  ideal as there is no loss to either privacy or utility through either approximate samples or a runtime side-channel. However, the downside of this method is that the acceptance probability for the current database must be known. Assuming that the proposal is normalized, this is equivalent to knowing the integrating constant for the target. While this may not be too cumbersome for low-dimensional settings, it becomes computationally intractable for high-dimensional distributions. In the next section, we give an alternative set of assumptions to remove the requirement of the integrating constant.

\begin{rem}\label{rem:wait}
A similar alternative to Theorem \ref{thm:memoryless} is as follows: during each step of the rejection sampler, if a sample is accepted, then with probability $c_D/c$ report the sample, and with probability $1-c_D/c$ do not report the sample. This results in the same runtime as Theorem \ref{thm:memoryless}. 

{\red We remark that this alternative algorithm has a similar flavor to the randomized response mechanism, one of the oldest privacy mechanisms \citep{warner1965randomized}. While beyond the scope of this paper, it may be worth investigating whether there is any deeper connection between this privacy-aware rejection sampler and randomized response. }
\end{rem}

 \begin{figure}
    \centering
    \includegraphics[width=.5\linewidth]{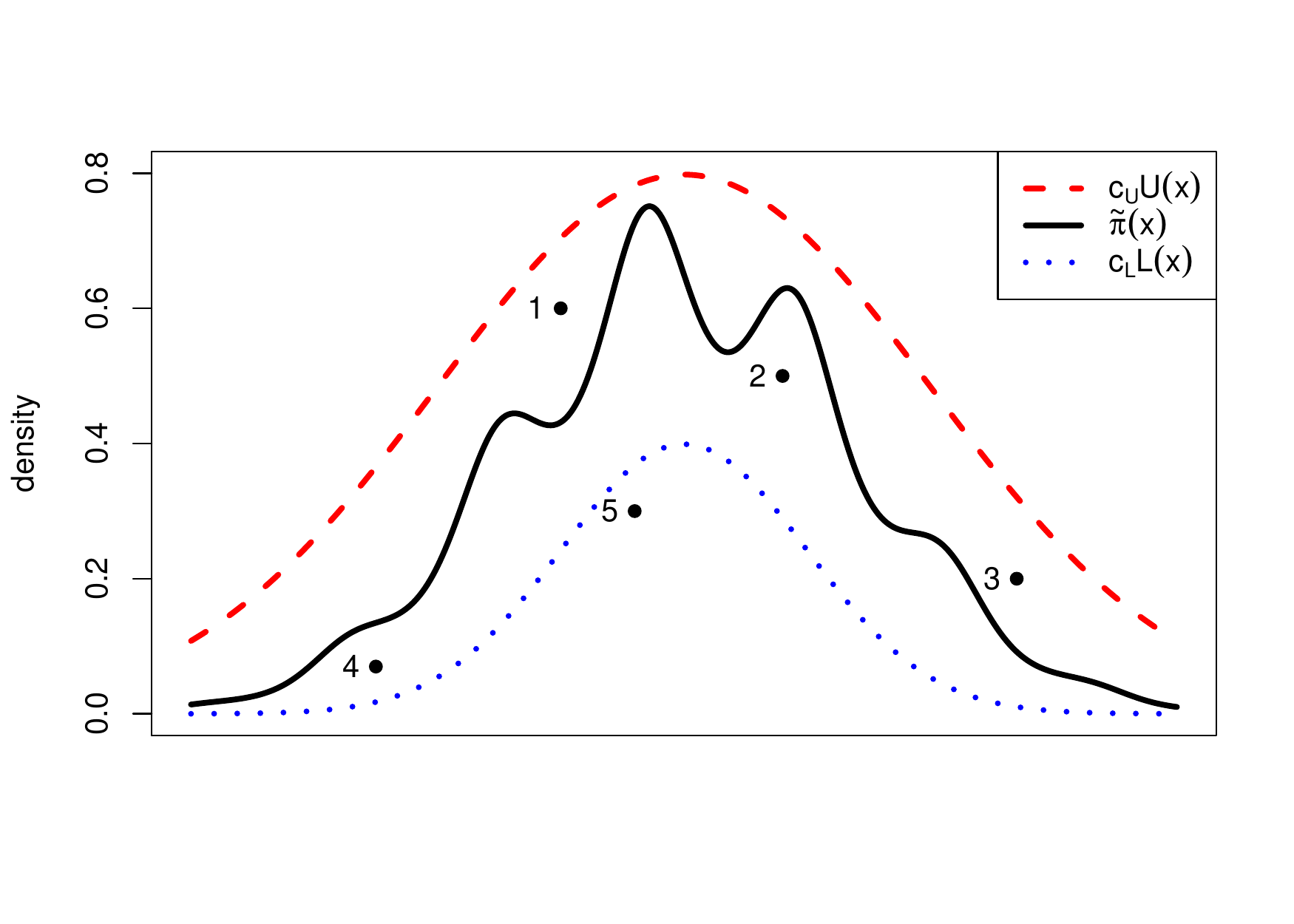}
    \caption{Example implementation of Algorithm \ref{alg:squeeze}. See Example \ref{ex:squeeze} for details. The sample $x_2$ is accepted, but not published until $x_5$.}
    \label{fig:squeeze}
\end{figure}
 \subsection{Implicit wait-time via squeeze function}\label{s:squeeze}
In this section, we propose another method of producing an exact rejection sampler with data-independent runtime. Our method, described in Algorithm \ref{alg:squeeze} and Theorem \ref{thm:squeeze}, avoids the need for the normalizing constant as in Theorem \ref{thm:memoryless} by instead using a carefully tailored {\em squeeze} function. {\red In the rejection sampling literature, a squeeze function is a lower bound on the target density which is assumed to be easy to evaluate, and which is used to speed up the computational time by avoiding evaluations of the target density when a proposed sample lies under the squeeze function (i.e.\ is rejected by the squeeze function,  see Example~\ref{ex:squeeze}). However, in this section, we will use the squeeze function not to speed up the computational time, but to slow it down; this will enable us to make the runtime equally distributed as in a worst-case setting. }

For this method, we assume that for each unnormalized target density $\twid\pi_D$, we have normalized densities $U_D(x)$ and $L_D(x)$ as well as constants $c_{L,D}$ and $c_{U,D}$ such that for all $x$ 
\[c_{L,D} L_D(x)\leq \twid\pi_D(x) \leq c_{U,D} U_D(x),\]
and such that the ratio $c_{L,D}/c_{U,D}$ does not depend on $D$. 
{\red Note that the latter condition is easy to enforce: if $c_{L,D}$ and $c_{U,D}$ are two valid constants, then so are $c^*_{L,D} < c_{L,D}$ and $c^*_{U,D}> c_{U,D}$.}
We then choose the value $X_s$ that we publish based on the rejection sampler that targets $\twid \pi_D$ from $U_D$, but do not publish the sample until a value is accepted from $L_D$ (i.e.\ the proposal lies under the squeeze function $c_{L,D}L_D$: see Example~\ref{ex:squeeze}). Because of this modification, the runtime is determined only by the ratio $c_L/c_U$, which is assumed to be constant across databases.   Thus, there is no additional privacy cost to using this sampler, since get an exact sample with runtime independent of $D$. This method is similar to that of Section \ref{s:additive} in that there is an additive wait-time, but Algorithm \ref{alg:squeeze} is able to do this implicitly, without knowing the acceptance probability for the current database. In Proposition \ref{prop:wait}, we show that the assumptions of Theorem \ref{thm:squeeze} are strictly weaker than those of Proposition \ref{thm:memoryless}.

 \begin{algorithm}
 \caption{Privacy-aware rejection sampling via squeeze functions}
 \scriptsize
 INPUT: $\twid\pi$, $U$, $L$, $c_U$, and $c_L$ such that $c_L L(x)\leq \twid \pi(x)\leq c_U U(x)$ for all $x$
 \begin{algorithmic}[1]
   \setlength\itemsep{0em}
   \STATE Set \texttt{anyAccepted}=\texttt{FALSE}
   \STATE Sample $X\sim U(x)$
   \STATE Sample $Y \sim \mathrm{Unif}(0,1)$ 
   \IF{$Y\leq \frac{\twid\pi(X)}{c_U U(X)}$ and \texttt{anyAccepted}==\texttt{FALSE}}
   \STATE Set $X_s=X$ 
   \STATE Set \texttt{anyAccepted}=\texttt{TRUE}
   \ENDIF
   \IF{$Y\leq \frac{c_LL(X)}{c_UU(X)}$}
   \STATE Publish $X_s$
   \ELSE
   \STATE Go to 2.
   \ENDIF
 \end{algorithmic}
 OUTPUT: $X_s$
 \label{alg:squeeze}
 \end{algorithm}

\begin{thm}\label{thm:squeeze}
Let $D\in \mscr D$ be a database and $\{\twid\pi_D\mid D\}$ be the (unnormalized) target densities. Assume that for each $\twid \pi_D$, we have normalized densities $U_D(x)$ and $L_D(x)$ as well as constants $c_{L,D}$ and $c_{U,D}$ such that the ratio $c_{L,D}/c_{U,D}$ does not depend on $D$ and such that 
for all $x$ $c_{L,D} L_D(x)\leq \twid\pi_D(x) \leq c_{U,D} U_D(x)$. Then the output of Algorithm \ref{alg:squeeze} with $\twid\pi=\twid\pi_D$, $U=U_D$, $L=L_D$, $c_{U} = c_{U,D}$, $c_L=c_{L,D}$ has distribution $\pi_D$ and runtime $\mathrm{Geom}(c_{L,D}/c_{U,D})$, which does not depend on $D$. 
\end{thm}
\begin{proof}
The published sample is determined by the condition $Y\leq \frac{\twid\pi_D(X)}{c_{U,D} U_D(X)}$, where $X\sim{\red U_D}(x)$ and $Y \sim \mathrm{Unif}(0,1)$. This is a simple rejection sampler, and so conditional on acceptance, $X\sim \pi_D$. However, a sample is not published until $Y\leq {\red \frac{c_{L,D}L_D(X)}{c_{U,D}U_D(X)}}$. This is a rejection sampler targeting $L_D(X)$, using the proposal $U_D(X)$ and threshold {\red $c_{U,D}/c_{L,D}$}. As such, the number of iterations is $\mathrm{Geom}(c_{L,D}/c_{U,D})$, which by assumption does not depend on $D$. 
\end{proof}

While the assumption of the squeeze functions in Theorem \ref{thm:squeeze} may seem unintuitive, it is in fact strictly weaker than knowing the integrating constant for $\twid \pi_D$, as was required in Section \ref{s:additive}, as shown in Proposition \ref{prop:wait}. In Section \ref{s:adaptive} and \ref{s:examples} we show that there are several natural instances of the exponential mechanism where the assumptions of Theorem \ref{thm:squeeze} are satisfied. 

\begin{prop}\label{prop:wait}
Let $D\in \mscr D$ be a database and $\{\pi_D\mid D\}$ be the normalized target densities. Assume that for each $\pi_D$, we have normalized densities $U_D(x)$ and constants $c_{U,D}$ such that $\pi_D(x)\leq c_{U,D} U_D(x)$. Choose a value $c\geq \sup_D c_{U,D}$. Then the squeeze function $L_D=\pi_D$, with constant $c_{L,D}=c_{U,D}/c$ satisfies the assumptions of Theorem \ref{thm:squeeze}, and the output of Algorithm \ref{alg:squeeze} has distribution $\pi_D$ and runtime $\mathrm{Geom}(1/c)$.
\end{prop}
\begin{proof}
Since $c\geq c_{U,D}$, we have that $c_{L,D}=c_{U,D}/c\leq 1$. So, $c_{U,D}L_D(x)\leq \pi_D(x)$ for all $x$. Then, the runtime of Algorithm \ref{alg:squeeze} is geometric with parameter $(c_{L,D}/c_{U,D})=1/c$, and the output of Algorithm \ref{alg:squeeze} has the appropriate distribution as argued in the proof of Theorem \ref{thm:squeeze}. 
\end{proof}

In fact, the application of Algorithm \ref{alg:squeeze} described in Proposition \ref{prop:wait} is very similar to the variation of Theorem \ref{thm:memoryless} described in Remark \ref{rem:wait}.

\begin{rem}
{\red 
Proposition \ref{prop:wait} showed that the assumptions for the the squeeze functions in Theorem \ref{thm:squeeze}, but are actually strictly weaker than the assumptions needed in Section \ref{s:additive}. Furthermore, it can be seen that the assumptions of Theorem \ref{thm:squeeze} (assuming that we can evaluate the constant $c_{L,D}/c_{U,D}$) are {\emph strictly stronger} than knowing the worst-case acceptance probability, which is needed for the truncated sampler of Section \ref{s:truncated} -- this is because the ratio  $c_{L,D}/c_{U,D}$ is itself a lower bound on the worst-case acceptance probability. }
\end{rem}

\begin{ex}\label{ex:squeeze}
Figure \ref{fig:squeeze} is an illustration of how Algorithm \ref{alg:squeeze} works. We see an example of a target $\twid \pi$, which satisfies $c_L L(x)\leq \twid \pi(x) \leq c_U U(x)$ for constants $c_L$, $c_U$, a proposal function $U$ and squeeze function $L$. The points $(x_i,y_i)$ are sequentially drawn uniformly within the area under $c_UU$; equivalently, $x_i \sim U(x)$ and $y_i=u_i\cdot c_U U(x)$, where $u_i \iid \mathrm{Unif}(0,1)$. Algorithm \ref{alg:squeeze} processes these samples as follows: For the first pair, $y_1>\twid \pi(x_1)$ so the sample is rejected. The second sample satisfies $y_2\leq \twid \pi(x_2)$ so it is accepted (set $X_s=x_2$), but because $y_2>c_L L(x_2)$ it is not published yet. The third sample is rejected since $y_3>\twid \pi(x_3)$. The fourth sample satisfies $y_4\leq \twid \pi(x_4)$, but since we already accepted $x_2$, we do not update $X_s$. Since $y_4> c_L L(x_4)$ we still do not publish anything yet. Finally, $y_5\leq c_L L(x_5)$ so we publish $X_s=x_2$. 

As noted in Theorem \ref{thm:squeeze}, the procedure results in $X_s\sim  \pi$, but the {\red runtime} is distributed as $\mathrm{Geom}(c_L/c_U)$, which does not directly depend on $\twid \pi$. 
\end{ex}
 \subsection{Adaptive rejection sampler for log-H\"{o}lder densities}\label{s:adaptive}
 The previous three subsections proposed modifications to simple rejection samplers in order to remove the runtime side-channel. In this section, we use the squeeze method of Section \ref{s:squeeze} to develop an adaptive rejection sampler with data-independent runtime for log-H\"{o}lder densities. Our method, outlined in Algorithm~\ref{alg:adapt}, is entirely black box, requiring only H\"{o}lder parameters $(s,H)$ that hold for every database, and is a modification of the (nearly) minimax optimal sampler of \citet{achddou2019minimax}. 
 Let $\pi_D(x) \propto \exp(\g_D(x))$ be an unnormalized  target density on a bounded convex set $C$, where $\g_D$ is $(s,H)$-H\"{o}lder for all datasets $D$: $|\g_D(x)-\g_D(y)|\leq H \lVert x-y\rVert^s$ for all $D$ and for all $x,y\in C$.  
 This setup differs from \citet{achddou2019minimax}, who assume that the target itself is H\"{o}lder, rather than the log-target. 
 This difference is important in order to derive upper and lower bounds that satisfy a property similar to Theorem \ref{thm:squeeze}. 
 We point out in Remark \ref{rem:holder} that the log-H\"{o}lder assumption, with the same $s$ and $H$ across all datasets, is natural for many privacy mechanisms, and many instances of the exponential mechanism in the literature satisfy this assumption. 
 
 {\red At a high-level, given evaluations of $g_D(x)$ at a finite set of locations, the log-H\"older assumption allows us to construct piecewise-constant upper and lower bounds on $g_D(x)$ and therefore the target density.
 Importantly, these bounds can be constructed so that the ratio of their associated normalization constants is independent of the database $D$.
 Then, in the fashion of Algorithm~\ref{alg:squeeze}, by proposing from the upper bound, and stopping only on accepting from the lower bound, we can have a database independent runtime.
 Following each proposal, we add a new location to our set of evaluations of $g_D(x)$, tightening the lower and upper bounds, and ensuring the acceptance probabability increases each iteration.
 We describe these steps in detail in Algorithm~\ref{alg:adapt}.
 }

  \begin{algorithm}
 \caption{Privacy-aware adaptive rejection}
 \scriptsize
 INPUT: $g$ an $(s,H)$-H\"{o}lder function on a bounded convex set $C\subset \RR^d$ for some norm $\lVert \cdot \rVert$, initial evaluation points $\{(x_1,g(x_1)),\ldots,(x_n,g(x_n))\}$, and a ``nearest neighbor'' map $P_T(\cdot):C\rightarrow T$ for any finite set $T\subset C$, the number $N$ of i.i.d. samples desired from $\pi(x)\propto \exp(g(x))I(x\in C)$
 \begin{algorithmic}[1]
   \setlength\itemsep{0em}
   \STATE Set {\red \texttt{anyAccepted}}=\texttt{FALSE},
   \texttt{numSamples}=0, and
   \texttt{publishedSamples}$=\emptyset$
   \STATE Set $S = \{(x_1,g(x_1)),\ldots,(x_n,g(x_n))\}$, and
   $T=\{x\mid (x,y)\in S \text{ for some } y\}$
   \WHILE{\texttt{numSamples}$<N$}
    \STATE Define $\hat g(x)=g(P_T(x))$ for all $x\in C$ {\red (note that this only requires evaluations of $g$ from $S$)}
    \STATE Set $\hat r\geq \sup_{x\in C} H\lVert x-P_T(x)\rVert^s$
    \STATE Sample $X\sim \exp(\hat g(x))/(\int_{C} \exp(\hat g(x)) \ dx)$
    \STATE Sample $Y\sim \mathrm{Unif}(0,1)$
    \IF{$Y\leq \exp(g(X))/\exp(\hat g(X)+\hat r)$ and  \texttt{anyAccepted}=\texttt{FALSE} }
    \STATE Set $X_s=X$
    \STATE Set  \texttt{anyAccepted}=\texttt{TRUE}
    \ENDIF
    \IF{$Y\leq \exp(-2\hat r)$}
    \STATE Publish $X_s$ and append $X_s$ to \texttt{publishedSamples}
    \STATE Increment  \texttt{numSamples} by 1
    \STATE Set  \texttt{anyAccepted}=\texttt{FALSE}
    \ENDIF
    \STATE Choose $Z\in C\setminus T$ either randomly or deterministically based on only $T$, $H$ and $s$
    \STATE Append $(Z,g(Z))$ to $S$
    \STATE Append $Z$ to $T$
    \ENDWHILE
    
     OUTPUT: \texttt{publishedSamples}, which can be published in a stream
 \end{algorithmic}
 \label{alg:adapt}
 \end{algorithm}
 
 \begin{thm}\label{thm:adapt}
 Let $\mscr D$ be a space of databases and $\{\twid \pi_D=\exp(g_D)\mid D\}$ be the unnormalized target densities, which have support on a bounded convex set $C$. Suppose that for all $D$, $g_D$ is $(s,H)$-H\"{o}lder with norm $\lVert\cdot \rVert$ on $C$. Then Algorithm \ref{alg:adapt} results in $N$ i.i.d. samples from $\twid \pi$ and has runtime between published samples which does not depend on $D$. If the mapping $P_T$ and the update procedure to generate $Z$ are chosen in a way that $\sup_{x\in C}\lVert x-P_T(x)\rVert\rightarrow 0$, then the probability of publishing an accepted sample in a given iteration converges to $1$. 
 \end{thm}
 \begin{proof}
 The quantity $\hat r$ is an upper bound on the maximum difference between $g_D(x)$ and $\hat g(x)$, by the H\"older assumption. So, at any point in the algorithm, since $\g_D$ is H\"{o}lder, and by the definition of $\hat r$, we have that 
     \[\exp(\hat g(x)-\hat r)\leq \exp(g_D(x))\leq \exp(\hat g(x)+\hat r).\]
     Using the notation of Theorem \ref{thm:squeeze}, we have that ${\red c_{L,D}}=\exp(-\hat r)$, ${\red L_D}(x) = k\exp(\hat g(x))$, ${\red c_{U,D}} = \exp(\hat r)$ and ${\red U_D}(x) = k\exp(\hat g(x))$, where $k=(\int_{C} \exp(\hat g(x)) \ dx)^{-1}$. Since {\red $c_{L,D}/c_{U,D}= \exp(-2\hat r)$ does not depend on $D$, }
     by Theorem \ref{thm:squeeze} the published samples are drawn independently from $\pi_D$ and the {\red runtime} does not depend on $D$. 
     
     The probability of publishing a sample is $\exp(-2\hat r)$. So, as long as $\lVert x-P_T(x)\rVert$ decreases as more samples $Z$ are appended to $T$, we have $\hat r\rightarrow 0$ and thus the probability of publishing an accepted sample converges to $1$. 
 \end{proof}


Because the update step and the rejection step are separated, we can think about the best way to update the proposal function. Our goal should be to reduce $\hat r$ as quickly as possible. A simple, but naive solution would be to sample $Z$ uniformly on $C$. Another approach would be to choose a sequence of $(x_i)_{i=1}^\infty$ such that for any $N$, the subset $(x_1)_{i=1}^N$ consists of approximately equally spaced points in $C$. This could be done intelligently using sequential space-filling experimental designs (e.g., \citealp{crombecq2010generating,pronzato2012design}). For example, a greedy maximin solution would be to choose $z = \arg\sup_{z\in C} H\lVert z-P_T(z)\rVert^s$ \citep{pronzato2012design}, which maximizes the publishing probability for the next iteration. Computing the maximin solution may be possible in low-dimensions, but becomes intractable in high dimensional spaces.

As in \citet{achddou2019minimax}, we can make the adaptive sampler much easier to implement by considering the following special case of Algorithm \ref{alg:adapt}: 1) use the $\ell_\infty$ norm in the H\"{o}lder definition, 2) set $C=[0,1]^d$, 3) approximate the nearest neighbor calculation $P_T(y)$ on a grid, as described in \citet[Definition 4]{achddou2019minimax}. These modifications make the construction, evaluation, and sampling of the proposal  $\exp(\hat g)$ computationally efficient, even in high dimensions. The accept-reject steps (lines 6-16) and the update steps (lines 17-19) can be done in batches to avoid updating the function $\hat g$ too often, when it will not significantly improve the acceptance probability. 

\begin{rem}[Relative runtime]
We consider how the runtime of Algorithm \ref{alg:adapt} compares to a similar sampler without the privacy constraint. Recall that the acceptance probability of our sampler is $\exp(-2\hat r)$. If we use the same proposal distribution, but base the acceptance criteria solely on the target, then the acceptance probability depends on the target. For a typical target density, we expect that the acceptance probability is approximately $\exp(-\hat r)$. If this is the case, then as $\hat r\rightarrow 0$, the ratio of the rejection probabilities is 
\[\lim_{\hat r\rightarrow 0} \frac{1-\exp(-2\hat r)}{1-\exp(-\hat r)}=\lim_{\hat r \rightarrow 0}\frac{2\hat r}{\hat r}=2,\]
where we use a series expansion of $\exp(-x)$ about zero to evaluate the limit. This suggests that the cost of privacy is that the rejection probability is about double that of the non-private sampler. 

Another cost of the privacy-preserving adaptive sampler is the decoupling of the rejection and update steps. Roughly, we will need to evaluate $\g_D$ twice as often---one for the update and one for the accept/reject step---as compared to non-private adaptive samplers, such as in \citet{achddou2019minimax}. This additional cost is somewhat mitigated by the fact that the update points can be chosen in a more intelligent manner, improving the rate of convergence of the proposal. 
\end{rem}

\begin{figure}[t]
    \centering
    \begin{minipage}{.48\linewidth}
    \includegraphics[width=\linewidth]{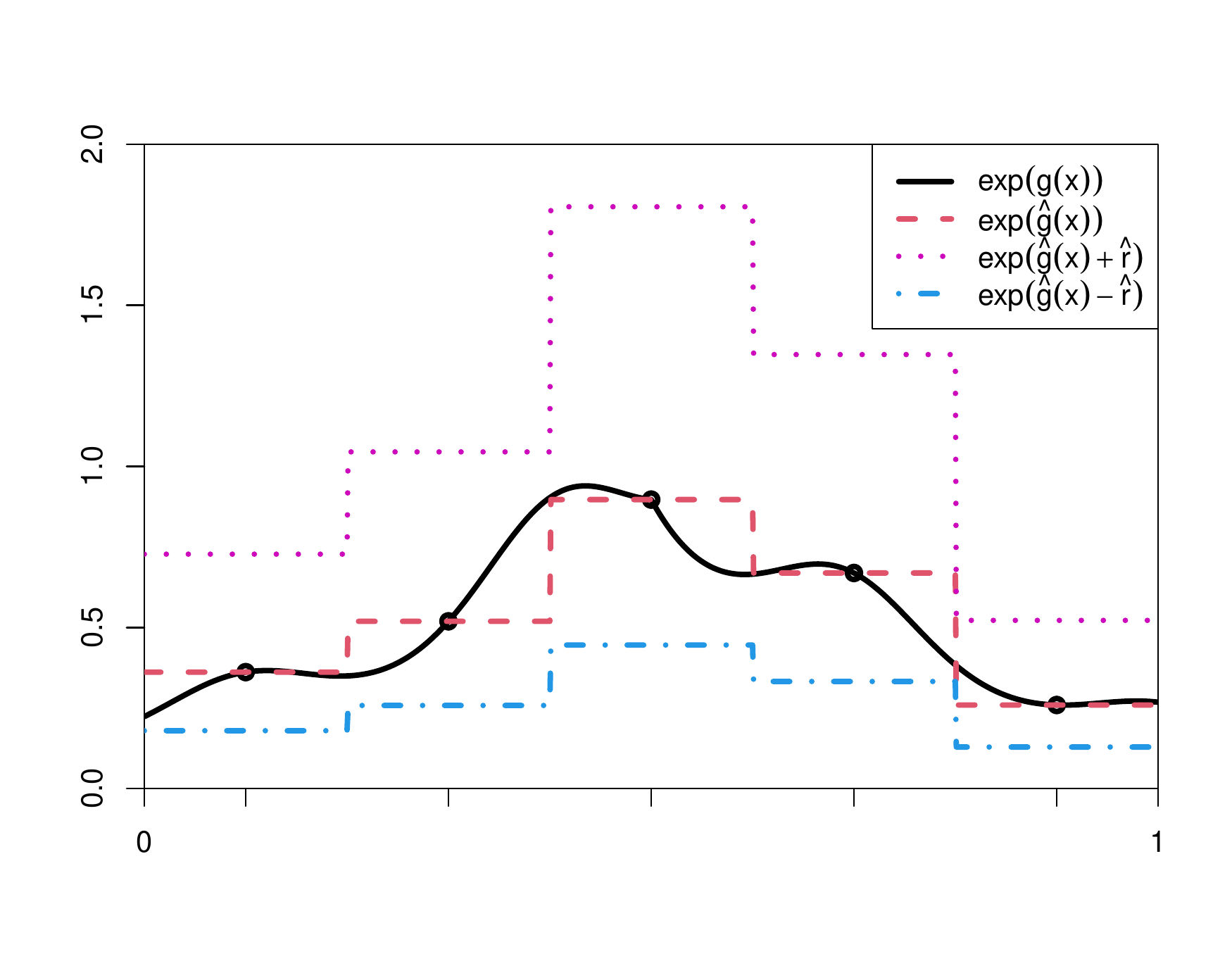}
        \vspace{-1.5cm}
    \end{minipage}
    \begin{minipage}{.48\linewidth}
    \includegraphics[width=\linewidth]{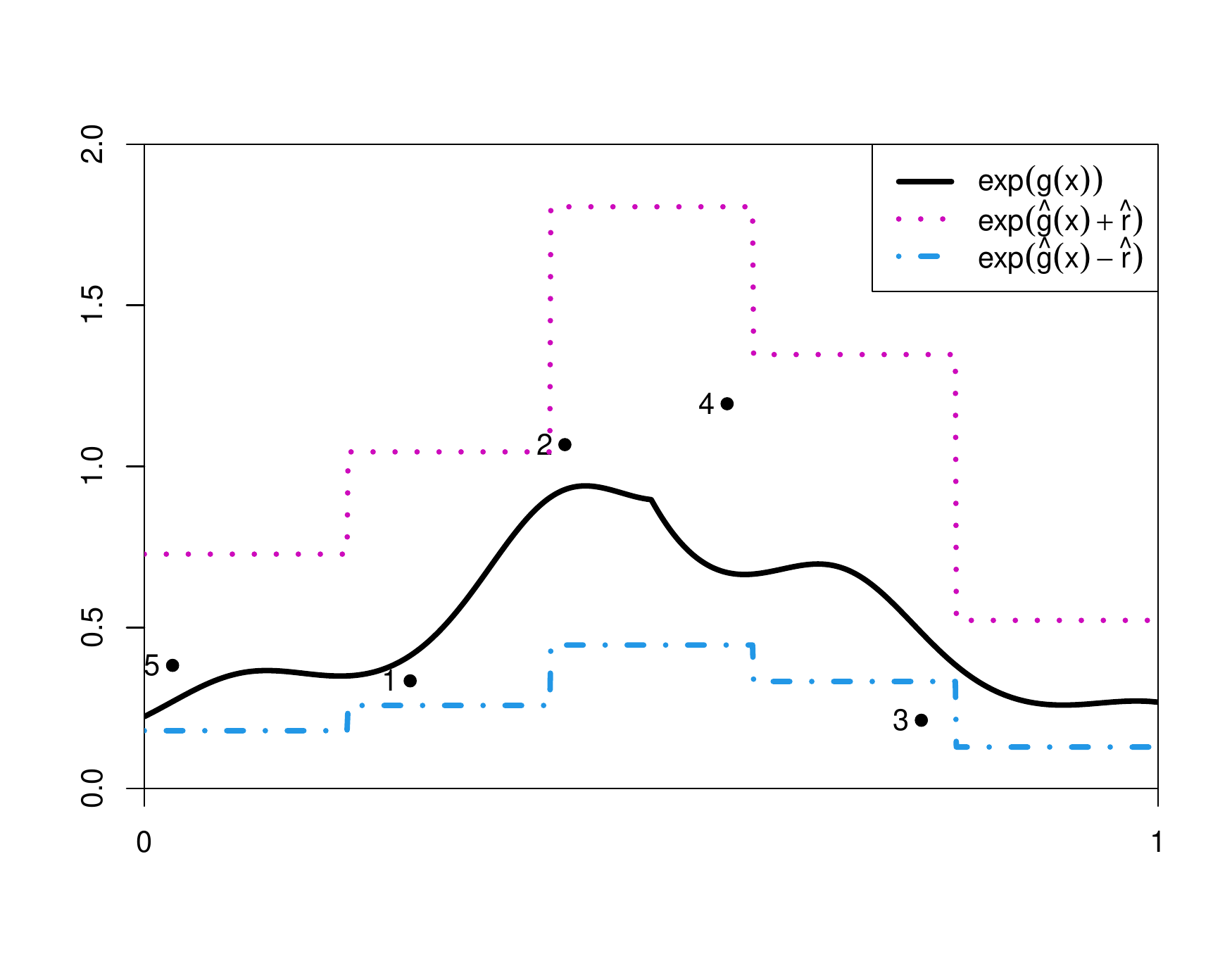}  
    \vspace{-1.5cm}
    \end{minipage}
    \begin{minipage}{.48\linewidth}
    \includegraphics[width=\linewidth]{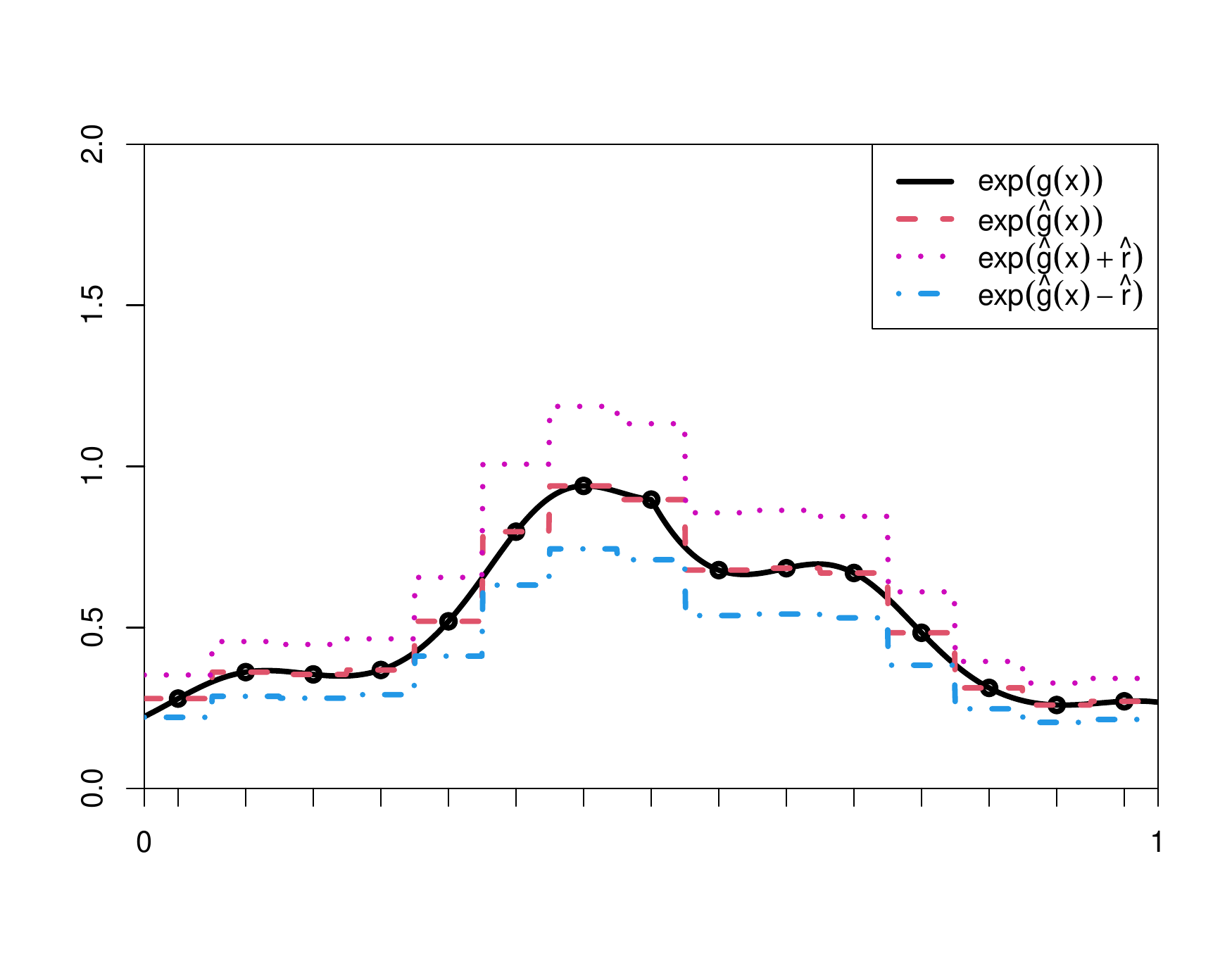}
        \vspace{-.5cm}
    \end{minipage}
    \begin{minipage}{.48\linewidth}
    \includegraphics[width=\linewidth]{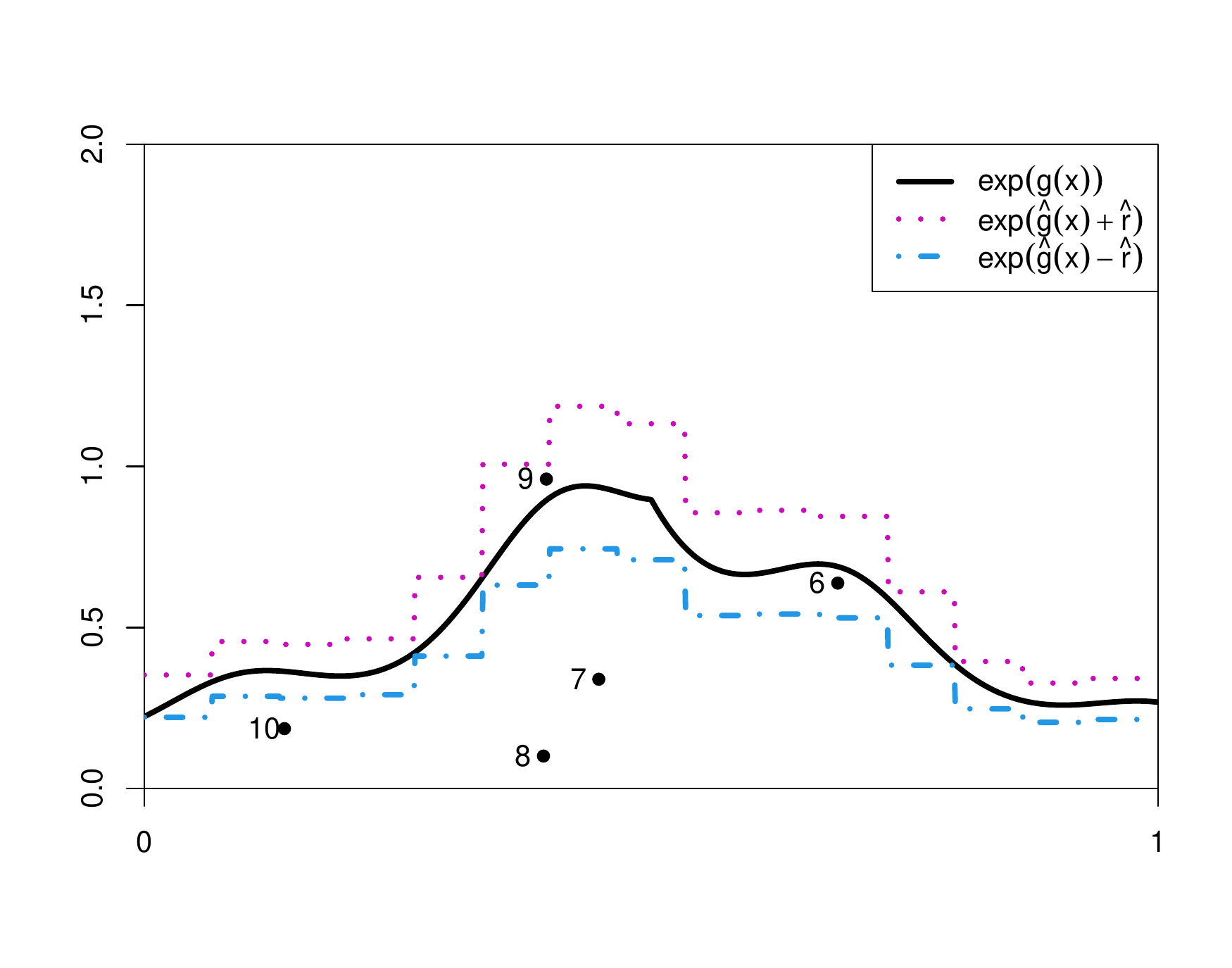}
        \vspace{-.5cm}
    \end{minipage}
    \caption{Implementation of Algorithm \ref{alg:adapt}, as explained in Example \ref{ex:adapt}. Plots progress in the normal reading order.}
    \label{fig:adapt}
\end{figure}

\begin{ex}\label{ex:adapt}
We illustrate Algorithm \ref{alg:adapt} applied to the target $\twid \pi = \exp(g(x))$, where  $g$ is the $7$-Lipschitz function $g(x) = -3|x-1/2|+(1/5)\sin(20x)$, so $H=7$ and $s=1$. The update and sampling steps of Algorithm \ref{alg:adapt} are run in batches. First, 5 equally spaces points are used to approximate $\exp(g(x))$ by a piece-wise linear function $\exp(\hat g(x))$. The upper bound is $\exp(\hat g(x)+\hat r)$ and the lower bound is $\exp(\hat g(x)-\hat r)$. The top left plot of Figure \ref{fig:adapt}, illustrates each of these functions. Next, five points $(x_i,y_i)$ for $i=1,\ldots, 5$ are sampled uniformly within the area under $\exp(\hat g(x)+\hat r)$, as seen in the top right plot of Figure \ref{fig:adapt}. The first value $y_1$ is below $\exp(g(x_1))$, but not below $\exp(\hat g(x_1)-\hat r)$, so we set $X_s=x_1$ and set $ \texttt{anyAccepted}=\texttt{TRUE}$, but do not publish $X_s$ yet. We reject $x_2$. Then as $y_3\leq \exp(\hat g(x_3)-\hat r)$, we publish $X_s=x_1$, and set $\texttt{anyAccepted}=\texttt{FALSE}$. We reject both $x_4$ and $x_5$. 

After this, we update the approximation $\hat g$ using 15 equally spaced points. This grid is a superset of the 5-point grid, so we can reuse the previous evaluations. The new approximation and bounds are shown in the bottom left plot of Figure \ref{fig:adapt}. Then $(x_i,y_i)$ for $i=6,\ldots, 10$ are sequentially sampled uniformly from the area under $\exp(\hat g(x)+\hat r)$, illustrated in the bottom right plot of Figure \ref{fig:adapt}. As $\exp(\hat g(x_6)-\hat r)<y_6\leq \exp(g(x_6))$, we set $X_s=x_6$ and $\texttt{anyAccepted}=\texttt{TRUE}$, but do not publish $X_s$. Since $y_7\leq \exp(\hat g(x_7)-\hat r)$, we publish $X_s=x_6$ at this time and set $\texttt{anyAccepted}=\texttt{FALSE}$. Then since $y_8\leq \exp(\hat g(x_8)-\hat r)$, we immediately publish $x_8$. We reject $x_9$. Last, as $y_{10}\leq \exp(\hat g(x_{10})-\hat r)$, we also publish $x_{10}$. 

Note that by using equally spaced points, $\hat r$ converges to zero rapidly, illustrating the benefit of using deterministically chosen points in the construction of $\hat g$. 
\end{ex}

\begin{rem}\label{rem:holder}
There are several prior DP works on the exponential mechanism, where the utility function is assumed to be Lipschitz (a special case of H\"{o}lder), and where Algorithm \ref{alg:adapt} can be applied. \citet{minami2016differential} assume Lipschitz and concave utility functions. 
\citet{bassily2014private} and \citet{bassily2014privateRevisited} derive optimal DP mechanisms under the assumption of Lipschitz and convex empirical risk objective functions, as well as a bounded domain, which result in implementations of the exponential mechanism. In part of their work, \citet{ganesh2020faster} assume Lipschitz and $L$-smooth utility functions in the exponential mechanism. 
\end{rem}

\section{Application to exponential mechanism sampling}\label{s:examples}
In this section, we explore some instances of the exponential mechansim that satisfy the assumptions of the rejection samplers proposed in Section \ref{s:results}, allowing for a privacy-preserving implementation.

\subsection{Strongly Concave and \texorpdfstring{$L$}\ -smooth log-density}\label{s:strongly}
We first consider instances of the exponential mechanism where the utility function $\g_D$ is both strongly concave and $L$-smooth. These are the same properties that \citet{ganesh2020faster} assume. 
Both \citet{awan2019benefits} and \citet{minami2016differential} assume strongly concave utility functions in the exponential mechanism. Other private empirical risk minimization works, while not working directly with the exponential mechansim, also commonly assume $L$-smooth and strong concavity \citep{kifer2012private,bassily2014private,bassily2014privateRevisited}.

Under the strongly concave and $L$-smooth assumptions, we are able to derive upper and lower bounds for the target, which satisfy the requirements of Theorem \ref{thm:squeeze}.
\begin{lem}\label{lem:strongly}
Let $\twid \pi_D(x)\propto\exp(\g_D(x))$ be the (unnormalized) target density, where $\g_D:\RR^d\rightarrow \RR$ is twice-differentiable, $\alpha$-strongly concave, and $L$-smooth. Call $x_D^*\defeq \argmax_x \g_D(x)$. Using $\phi_d(x;\mu,\Sigma)$ to denote the pdf of $N_d(\mu,\Sigma)$, we have for all $x$,
\[\exp(\g_D(x_D^*)) \left(2\pi/L\right)^{d/2}\phi_d(x;x_D^*,L^{-1}I)\leq \exp(\g_D(x))\leq \exp(\g_D(x_D^*)) \left(2\pi/\alpha\right)^{d/2} \phi_d(x;x_D^*,\alpha^{-1}I).\]

Furthermore, calling ${\red c_{L,D}} = \exp(\g_D(x_D^*)) \left(2\pi/L\right)^{d/2}$ and ${\red c_{U,D}}=\exp(\g_D(x_D^*)) \left(2\pi/\alpha\right)^{d/2}$, we have that ${\red c_{L,D}/c_{U,D}} = (\alpha/L)^{d/2}$, which does not depend on $D$.
\end{lem}
\begin{proof}
By strong concavity, we have that 
\[    {\red -\g_D(x)\geq  -g_D(x^*) - \nabla \g_D(x_D^*)^\top (x-x_D^*) + \frac{\alpha}{2} \lVert x_D^*-x\rVert_2^2=\g_D(x_D^*) + \frac{\alpha}{2} \lVert x_D^*-x\rVert_2^2,}
\]
since $\nabla \g_D(x_D^*)=0$. This implies that \[\exp(\g_D(x))\leq \exp(\g_D(x_D^*)) \exp\left(-\frac{\lVert x_D^*-x\rVert_2^2}{2(1/\alpha)}\right).\] 
Including the integrating constant for a multivariate normal distribution gives the upper bound.

Next, since $\g_D$ is $L$-smooth, we have that 
\[
    {\red -\g_D(x)= -\g_D(x_D^*) - \nabla \g_D(x_D^*)^\top (x-x_D^*) + \frac 12 (x_D^*-x)^\top\nabla^2 \g_D(\twid x)(x_D^*-x)\leq -\g_D(x_D^*) + \frac{L}{2}\lVert x_D^*-x\rVert_2^2,}
\]
where $\twid x$ is between $x_D^*$ and $x$, we used the fact that $\nabla \g_D(x_D^*)=0$, and that the eigenvalues of $\nabla^2 \g_D$ are upper bounded by $L$. This implies that 
\[\exp(\g_D(x))\geq \exp(\g_D(x_D^*)) \left(\frac{2\pi}{L}\right)^{d/2} \phi_d\left(x;x_D^*,L^{-1}I\right),\]
giving the lower bound.
\end{proof}

Given the bounds in Lemma \ref{lem:strongly}, we can now implement the squeeze-function rejection sampler of Section \ref{s:squeeze}, since {\red $c_{L,D}/c_{U,D}$} does not depend on $D$. As discussed in Proposition \ref{prop:wait}, generating these bounds is strictly easier than computing the integrating constant for the target, which is not needed in Lemma \ref{lem:strongly}. 

We could also implement the truncated sampler of Section \ref{s:truncated}, by using the bound $(\alpha/L)^{d/2}$ on the worst case acceptance probability. However, since Theorem \ref{thm:squeeze} is applicable, the truncated sampler is strictly worse as it incurs a price both in privacy as well as utility, whereas the squeeze sampler produces perfect samples.

There are many natural problem settings that fit the assumptions of Lemma \ref{lem:strongly}, particularly in empirical risk minization.

\begin{ex}[Strongly convex empirical risk minimization]\label{ex:erm}
Suppose that the database can be written as a vector $D = (d_1,\ldots, d_n)$, where $d_i$ is the contribution of individual $i$. Take as our utility function $\g_D(x) = -\left(\sum_{i=1}^n \ell(x;d_i)+r(x)\right)$, where $\ell(x;d)$ is a twice-differentiable convex function which is $L$-smooth and satisfies $\sup_{d,d'}\sup_x |\ell(x;d)-\ell(x;d')|\leq \Delta$, and $r(x)$ is an $\alpha$-strongly convex regularizer, which does not depend on the database $D$. For instance, we could take $r(x) = \frac{\alpha}{2}\lVert x\rVert_2^2$. Then the exponential mechanism samples from $\pi_D(x) \propto \exp(\frac{\ep}{2\Delta} \g_D(x))$ and satisfies $\ep$-DP. 

Note that $\g_D$ is $nL$-smooth and $\alpha$-strongly concave for all $D$, so it fits the framework of Lemma \ref{lem:strongly}. Such a setup is common in private empirical risk minmization \citep{kifer2012private,bassily2014private,bassily2014privateRevisited}, and in particular for private regression problems \citep{kifer2012private,reimherr2019kng,awan2020structure}.
\end{ex}


\subsection{KNG/Gradient Mechanism}\label{s:kng}

An alternative to simply applying the exponential mechanism to a strongly concave utility function is the $K$-norm gradient mechanism (KNG), proposed in  \citet{reimherr2019kng}, also known as the gradient mechanism \citep{asi2020instance}. KNG has been applied to applications such as geometric median estimation, and linear and quantile regression \citep{reimherr2019kng,asi2020instance}. Given an objective function $\g_D(x)$, KNG samples from $\pi_D(x)\propto \exp(-\frac{\ep}{2\Delta} \lVert \nabla \g_D(x) \rVert_K)$, where  $\Delta \geq \sup_{d(D,D')\leq1}\sup_x \lVert \nabla \g_D(x)-{\red \nabla }\g_{D'}(x)\rVert_K$, and where $\lVert \cdot \rVert_K$ is a chosen norm. 

While the exponential mechanism with a strongly concave utility is naturally approximated by a Gaussian distribution \citep{awan2019benefits}, KNG is closely related to the $K$-norm distributions \citep{reimherr2019kng}. The $K$-norm mechanism was introduced in \citet{hardt2010geometry}, and were also studied in \citet{awan2020structure}. 

\begin{defn}[$K$-norm distribution: \citet{hardt2010geometry}]
Let $\lVert \cdot \rVert_K$ be a norm on $\RR^d$, with associated unit norm ball: $K = \{x\in \RR^D\mid \lVert x \rVert_K\leq 1\}$. The $K$-norm distribution with location $m$ and scale $s$ has density 
\[f(x;m,s) = c^{-1} \exp\left(-s^{-1}\lVert x-m\rVert_K\right),\]
where $c = (d!) s^d\mathrm{Vol}(K)$.
\end{defn}

Under similar assumptions as those in \citet[Theorem 3.1]{reimherr2019kng}, Lemma \ref{lem:kngUniversal} gives upper and lower bounds which satisfy the assumptions required for Theorem \ref{thm:squeeze}.

\begin{lem}\label{lem:kngUniversal}
Let $\twid \pi_D(x)=\exp(-\lVert \nabla \g_D(x)\rVert_2)$ be the unnormalized target density, where $\g_D:\RR^d\rightarrow \RR$ is twice-differentiable, $\alpha$-strongly convex, and $L$-smooth. Call $x_D^*\defeq \argmin_x \g_D(x)$. Write $\psi_d(x;m,s)$ to denote the pdf of a $d$-dimensional $K$-norm distribution with location $m$, scale $s$, and $\ell_2$ norm. Denote $\mathrm{Vol}_d(\ell_2) = \frac{2^d \Gamma^d(1+1/2)}{\Gamma(1+d/2)}$ the volume of the unit $\ell_2$ ball in $\RR^d$. Then for all $x$, 
\[ (d!)L^{-d}\mathrm{Vol}_d(\ell_2) \psi_d(x;x_D^*,1/L)\leq \exp(-\lVert \nabla \g_D(x)\rVert_2)
\leq (d!)\alpha^{-d}\mathrm{Vol}_d(\ell_2) \psi_d(x;x_D^*,1/\alpha).\]

Furthermore, calling ${\red c_{L,D}} =(d!)L^{-d}\mathrm{Vol}_d(\ell_2) $ and ${c_{U,D}}=(d!)\alpha^{-d}\mathrm{Vol}_d(\ell_2)$, we have that ${\red c_{L,D}/c_{U,D}} = (\alpha/L)^{d}$, which does not depend on $D$.
\end{lem}
\begin{proof}
By strong convexity, we have that 
\begin{align*}
    \alpha \lVert{\red x- x_D^*}\rVert_2^2&\leq \langle \nabla \g_D(x)-\nabla \g_D(x_D^*),x-x_D^*\rangle=\langle \nabla \g_D(x),x-x_D^*\rangle
    \leq \lVert \nabla \g_D(x)\rVert_2 \cdot  \lVert x-x_D^*\rVert_2,
\end{align*}
where we used the fact that $\nabla \g_D(x_D^*)=0$ and Cauchy-Schwartz inequality. This implies that $\lVert \nabla \g_D(x) \rVert_2 \geq \alpha \lVert x-x_D^*\rVert_2$, which gives the upper bound. 

Next, as $\g_D$ is $L$-smooth, we have that 
\[ \lVert \nabla \g_D(x)\rVert_2 = \lVert \nabla \g_D(x)-\nabla \g_D(x_D^*)\rVert_2 \leq L\lVert x-x_D^*\rVert_2,\]
which gives the lower bound. 
\end{proof}

Lemma \ref{lem:kngUniversal} provides bounds that can be used to implement the sampler of Section \ref{s:squeeze}. The acceptance probability when targeting the lower bound is $(\alpha/L)^{d}$, which is independent of $D$, as required. 
While we could implement the sampler of Section \ref{s:truncated}, as  $(\alpha/L)^{d}$ is a bound on the worst case acceptance probability, this method is strictly worse than the squeeze sampler, as discussed in Section \ref{s:strongly}. 
The empirical risk problems of Example \ref{ex:erm} are also applicable to KNG, and offer several natural instances that satisfy the assumptions of Lemma \ref{lem:kngUniversal}.

Finally, note that for the KNG mechanism, if the underlying utility is $L$-smooth (not necessarily strongly concave), then the log-density is $L$-Lipschitz. As such, we can apply the adaptive rejection sampler of Section \ref{s:adaptive}. If multiple i.i.d. samples are required, this can provide a very computationally efficient sampling method, while keeping the runtime data-independent.

\section{Discussion}

In this paper, we first characterized the privacy cost due to the {\red runtime} of both simple and adaptive rejection samplers in terms of $\ep$-DP, $(\ep,\de)$-DP, and $f$-DP. We found that the runtime of standard rejection samplers can result in an arbitrary increase in the privacy cost, motivating the need for privacy-aware samplers. We then proposed three novel modifications to simple rejection samplers with varying assumptions, which all resulted in data-independent runtime. We also developed a privacy-aware adaptive rejection sampler for log-H\"{o}lder densities.

{\red There are three factors that influence the practicality of our algorithms, (1) {\em  the scalability of rejection sampling:} Typically, the acceptance probability of a rejection samplers decays exponentially with data dimension, making them impractical for very high dimensional problems. However, imposing additional structure like log-concavity or log-H\"older on the target density, adaptive rejection samplers (like our proposed one) can be applicable to higher-dimensional problems. 
 Such structural assumptions, as well as low- to moderate-dimensional problems are common in differential privacy applications.
(2) {\em the additional cost of our differentially-private modifications of rejection sampling:}
Our algorithms typically result in a reduction in the acceptance probability to match the worst-case dataset. This is unavoidable. However, our adaptive rejection sampler does not require any knowledge of this worst-case database. As such, the sample complexity of the runtime is the same as for a regular rejection sampler, but where the acceptance probability is the worst case.
(3) {\em The additional book-keeping overhead in implementing our differentially private rejection samplers:} All of our algorithms are minor modifications of existing simple or adaptive rejection samplers, and as such, this overhead is minimal.}

Of our proposed modifications to the rejection sampler, the squeeze method of Section \ref{s:squeeze} is the most powerful. We showed in Section \ref{s:examples} that for many instances of the exponential mechanism, appropriate upper and lower bounds can be generated. Furthermore, our adaptive sampling scheme is also built on the squeeze sampler. In a way, Algorithm \ref{alg:squeeze} can be viewed as a coupling of the sampler applied to the present database and a worst case database. It is an open question whether similar couplings could be developed with even weaker assumptions. 

An alternative to rejection sampling is \emph{coupling from the past} (CFP) \citep{propp1998coupling}, a modified MCMC approach. The benefit of CFP is that it is another perfect sampler, and could be a useful technique in designing privacy-aware samplers. The techniques used in this paper may be useful for determining the privacy cost of timing channel attacks on CFP and developing CFP algorithms with data-independent runtime. A variation on CFP is \emph{perfect tempering}, which also results in perfect samplers \citep{moller1999perfect,daghofer2004perfect,brooks2006perfect}, and may be an another approach to developing privacy-aware samplers.

{\red While in this paper we developed samplers whose runtime does not depend on the dataset, one could instead ask for the runtime to be differentially private by itself. Theorem \ref{thm:privacyCost} shows that a naive rejection sampler does have an inherent privacy cost, but one could also imagine altering the runtime to give a stronger privacy guarantee. A significant challenge with this approach is that we can only increase, but not reduce, the runtime. Due to this constraint, many existing DP techniques are not applicable. We leave it for future research to investigate mechanisms to privatize the runtime.}

\section*{Acknowledgments}
{\red This work was supported in part by the National Science Foundation (NSF) grants SES 2150615, DMS 1812197 and IIS 1816499. }

\bibliography{Bibliography.bib}

\begin{thebibliography}{58}
\providecommand{\natexlab}[1]{#1}
\providecommand{\url}[1]{\texttt{#1}}
\expandafter\ifx\csname urlstyle\endcsname\relax
  \providecommand{\doi}[1]{doi: #1}\else
  \providecommand{\doi}{doi: \begingroup \urlstyle{rm}\Url}\fi

\bibitem[Abowd(2018)]{abowd2018us}
John~M Abowd.
\newblock The {US Census Bureau} adopts differential privacy.
\newblock In \emph{Proceedings of the 24th ACM SIGKDD International Conference
  on Knowledge Discovery \& Data Mining}, pages 2867--2867, 2018.

\bibitem[Achddou et~al.(2019)Achddou, Lam-Weil, Carpentier, and
  Blanchard]{achddou2019minimax}
Juliette Achddou, Joseph Lam-Weil, Alexandra Carpentier, and Gilles Blanchard.
\newblock A minimax near-optimal algorithm for adaptive rejection sampling.
\newblock In \emph{Algorithmic Learning Theory}, pages 94--126. PMLR, 2019.

\bibitem[Asi and Duchi(2020{\natexlab{a}})]{asi2020instance}
Hilal Asi and John~C Duchi.
\newblock Instance-optimality in differential privacy via approximate inverse
  sensitivity mechanisms.
\newblock In H.~Larochelle, M.~Ranzato, R.~Hadsell, M.~F. Balcan, and H.~Lin,
  editors, \emph{Advances in Neural Information Processing Systems}, volume~33,
  pages 14106--14117. Curran Associates, Inc., 2020{\natexlab{a}}.

\bibitem[Asi and Duchi(2020{\natexlab{b}})]{asi2020near}
Hilal Asi and John~C Duchi.
\newblock Near instance-optimality in differential privacy.
\newblock \emph{arXiv preprint arXiv:2005.10630}, 2020{\natexlab{b}}.

\bibitem[Awan and Slavkovi{\'c}(2020)]{awan2020structure}
Jordan Awan and Aleksandra Slavkovi{\'c}.
\newblock Structure and sensitivity in differential privacy: Comparing k-norm
  mechanisms.
\newblock \emph{Journal of the American Statistical Association}, pages 1--20,
  2020.

\bibitem[Awan et~al.(2019)Awan, Kenney, Reimherr, and
  Slavkovi{\'c}]{awan2019benefits}
Jordan Awan, Ana Kenney, Matthew Reimherr, and Aleksandra Slavkovi{\'c}.
\newblock Benefits and pitfalls of the exponential mechanism with applications
  to {Hilbert} spaces and functional pca.
\newblock In \emph{International Conference on Machine Learning}, pages
  374--384. PMLR, 2019.

\bibitem[Bassily et~al.(2014{\natexlab{a}})Bassily, Smith, and
  Thakurta]{bassily2014private}
Raef Bassily, Adam Smith, and Abhradeep Thakurta.
\newblock Private empirical risk minimization: Efficient algorithms and tight
  error bounds.
\newblock In \emph{2014 IEEE 55th Annual Symposium on Foundations of Computer
  Science}, pages 464--473. IEEE, 2014{\natexlab{a}}.

\bibitem[Bassily et~al.(2014{\natexlab{b}})Bassily, Smith, and
  Thakurta]{bassily2014privateRevisited}
Raef Bassily, Adam Smith, and Abhradeep Thakurta.
\newblock Private empirical risk minimization, revisited.
\newblock \emph{rem}, 3:\penalty0 19, 2014{\natexlab{b}}.

\bibitem[Brooks et~al.(2006)Brooks, Fan, and Rosenthal]{brooks2006perfect}
Stephen~P Brooks, Yanan Fan, and Jeffrey~S Rosenthal.
\newblock Perfect forward simulation via simulated tempering.
\newblock \emph{Communications in Statistics-Simulation and Computation},
  35\penalty0 (3):\penalty0 683--713, 2006.

\bibitem[Bun and Steinke(2016)]{bun2016concentrated}
Mark Bun and Thomas Steinke.
\newblock Concentrated differential privacy: Simplifications, extensions, and
  lower bounds.
\newblock In \emph{Theory of Cryptography Conference}, pages 635--658.
  Springer, 2016.

\bibitem[Bun et~al.(2018)Bun, Dwork, Rothblum, and Steinke]{bun2018composable}
Mark Bun, Cynthia Dwork, Guy~N Rothblum, and Thomas Steinke.
\newblock Composable and versatile privacy via truncated cdp.
\newblock In \emph{Proceedings of the 50th Annual ACM SIGACT Symposium on
  Theory of Computing}, pages 74--86, 2018.

\bibitem[Chaudhuri et~al.(2012)Chaudhuri, Sarwate, and
  Sinha]{chaudhuri2012near}
Kamalika Chaudhuri, Anand Sarwate, and Kaushik Sinha.
\newblock Near-optimal differentially private principal components.
\newblock \emph{Advances in Neural Information Processing Systems},
  25:\penalty0 989--997, 2012.

\bibitem[Chaudhuri et~al.(2013)Chaudhuri, Sarwate, and
  Sinha]{chaudhuri2013near}
Kamalika Chaudhuri, Anand~D Sarwate, and Kaushik Sinha.
\newblock A near-optimal algorithm for differentially-private principal
  components.
\newblock \emph{Journal of Machine Learning Research}, 14, 2013.

\bibitem[Crombecq and Dhaene(2010)]{crombecq2010generating}
Karel Crombecq and Tom Dhaene.
\newblock Generating sequential space-filling designs using genetic algorithms
  and {Monte Carlo} methods.
\newblock In \emph{Asia-Pacific Conference on Simulated Evolution and
  Learning}, pages 80--84. Springer, 2010.

\bibitem[Daghofer and von~der Linden(2004)]{daghofer2004perfect}
Maria Daghofer and Wolfgang von~der Linden.
\newblock Perfect tempering.
\newblock In \emph{AIP Conference Proceedings}, volume 735, pages 355--362.
  American Institute of Physics, 2004.

\bibitem[Dimitrakakis et~al.(2017)Dimitrakakis, Nelson, Zhang, Mitrokotsa, and
  Rubinstein]{dimitrakakis2017differential}
Christos Dimitrakakis, Blaine Nelson, Zuhe Zhang, Aikateirni Mitrokotsa, and
  Benjamin Rubinstein.
\newblock Differential privacy for {Bayesian} inference through posterior
  sampling.
\newblock \emph{Journal of Machine Learning Research}, 18\penalty0
  (11):\penalty0 1--39, 2017.

\bibitem[Ding et~al.(2017)Ding, Kulkarni, and Yekhanin]{ding2017collecting}
Bolin Ding, Janardhan Kulkarni, and Sergey Yekhanin.
\newblock Collecting telemetry data privately.
\newblock In \emph{Advances in Neural Information Processing Systems 30},
  December 2017.

\bibitem[Dodis et~al.(2012)Dodis, L{\'o}pez-Alt, Mironov, and
  Vadhan]{dodis2012differential}
Yevgeniy Dodis, Adriana L{\'o}pez-Alt, Ilya Mironov, and Salil Vadhan.
\newblock Differential privacy with imperfect randomness.
\newblock In \emph{Annual Cryptology Conference}, pages 497--516. Springer,
  2012.

\bibitem[Dong et~al.(2022)Dong, Roth, and Su]{dong2022gaussian}
Jinshuo Dong, Aaron Roth, and Weijie~J Su.
\newblock Gaussian differential privacy.
\newblock \emph{Journal of the Royal Statistical Society Series B}, 84\penalty0
  (1):\penalty0 3--37, 2022.

\bibitem[Dwork et~al.(2006)Dwork, McSherry, Nissim, and
  Smith]{dwork2006calibrating}
Cynthia Dwork, Frank McSherry, Kobbi Nissim, and Adam Smith.
\newblock Calibrating noise to sensitivity in private data analysis.
\newblock In \emph{Theory of Cryptography Conference}, pages 265--284.
  Springer, 2006.

\bibitem[Dwork et~al.(2014)Dwork, Roth, et~al.]{dwork2014algorithmic}
Cynthia Dwork, Aaron Roth, et~al.
\newblock The algorithmic foundations of differential privacy.
\newblock \emph{Foundations and Trends in Theoretical Computer Science},
  9\penalty0 (3-4):\penalty0 211--407, 2014.

\bibitem[Erlingsson et~al.(2014)Erlingsson, Pihur, and
  Korolova]{erlingsson2014rappor}
{\'U}lfar Erlingsson, Vasyl Pihur, and Aleksandra Korolova.
\newblock Rappor: Randomized aggregatable privacy-preserving ordinal response.
\newblock In \emph{Proceedings of the 2014 ACM SIGSAC conference on computer
  and communications security}, pages 1054--1067, 2014.

\bibitem[Feldman et~al.(2009)Feldman, Fiat, Kaplan, and
  Nissim]{feldman2009private}
Dan Feldman, Amos Fiat, Haim Kaplan, and Kobbi Nissim.
\newblock Private coresets.
\newblock In \emph{Proceedings of the forty-first annual ACM symposium on
  Theory of computing}, pages 361--370, 2009.

\bibitem[Ganesh and Talwar(2020)]{ganesh2020faster}
Arun Ganesh and Kunal Talwar.
\newblock Faster differentially private samplers via {R\'{e}nyi} divergence
  analysis of discretized {Langevin MCMC}.
\newblock In H.~Larochelle, M.~Ranzato, R.~Hadsell, M.~F. Balcan, and H.~Lin,
  editors, \emph{Advances in Neural Information Processing Systems}, volume~33,
  pages 7222--7233. Curran Associates, Inc., 2020.

\bibitem[Garfinkel and Leclerc(2020)]{garfinkel2020randomness}
Simson~L Garfinkel and Philip Leclerc.
\newblock Randomness concerns when deploying differential privacy.
\newblock In \emph{Proceedings of the 19th Workshop on Privacy in the
  Electronic Society}, pages 73--86, 2020.

\bibitem[Haeberlen et~al.(2011)Haeberlen, Pierce, and
  Narayan]{haeberlen2011differential}
Andreas Haeberlen, Benjamin~C Pierce, and Arjun Narayan.
\newblock Differential privacy under fire.
\newblock In \emph{USENIX Security Symposium}, volume~33, 2011.

\bibitem[Hardt and Talwar(2010)]{hardt2010geometry}
Moritz Hardt and Kunal Talwar.
\newblock On the geometry of differential privacy.
\newblock In \emph{Proceedings of the forty-second ACM symposium on Theory of
  computing}, pages 705--714, 2010.

\bibitem[Hundepool et~al.(2012)Hundepool, Domingo-Ferrer, Franconi, Giessing,
  Nordholt, Spicer, and De~Wolf]{hundepool2012statistical}
Anco Hundepool, Josep Domingo-Ferrer, Luisa Franconi, Sarah Giessing,
  Eric~Schulte Nordholt, Keith Spicer, and Peter-Paul De~Wolf.
\newblock \emph{Statistical disclosure control}.
\newblock John Wiley \& Sons, 2012.

\bibitem[Ilvento(2020)]{ilvento2020implementing}
Christina Ilvento.
\newblock Implementing the exponential mechanism with base-2 differential
  privacy.
\newblock In \emph{Proceedings of the 2020 ACM SIGSAC Conference on Computer
  and Communications Security}, pages 717--742, 2020.

\bibitem[Joy~Persial et~al.(2011)Joy~Persial, Prabhu, and
  Shanmugalakshmi]{joy2011side}
G~Joy~Persial, M~Prabhu, and R~Shanmugalakshmi.
\newblock Side channel attack-survey.
\newblock \emph{International Journal of Advanced Scientific Research and
  Review}, 1\penalty0 (4):\penalty0 54--57, 2011.

\bibitem[Kapralov and Talwar(2013)]{kapralov2013differentially}
Michael Kapralov and Kunal Talwar.
\newblock On differentially private low rank approximation.
\newblock In \emph{Proceedings of the twenty-fourth annual ACM-SIAM symposium
  on Discrete algorithms}, pages 1395--1414. SIAM, 2013.

\bibitem[Kifer et~al.(2012)Kifer, Smith, and Thakurta]{kifer2012private}
Daniel Kifer, Adam Smith, and Abhradeep Thakurta.
\newblock Private convex empirical risk minimization and high-dimensional
  regression.
\newblock In \emph{Conference on Learning Theory}, pages 25--1. JMLR Workshop
  and Conference Proceedings, 2012.

\bibitem[Lane et~al.(2014)Lane, Stodden, Bender, and
  Nissenbaum]{lane2014privacy}
Julia Lane, Victoria Stodden, Stefan Bender, and Helen Nissenbaum.
\newblock \emph{Privacy, big data, and the public good: Frameworks for
  engagement}.
\newblock Cambridge University Press, 2014.

\bibitem[Leydold et~al.(2002)Leydold, Janka, and
  H{\"o}rmann]{leydold2002variants}
Josef Leydold, Erich Janka, and Wolfgang H{\"o}rmann.
\newblock Variants of transformed density rejection and correlation induction.
\newblock In \emph{{Monte Carlo} and Quasi-{Monte Carlo} Methods 2000}, pages
  345--356. Springer, 2002.

\bibitem[Martino(2018)]{martino2018independent}
Luca Martino.
\newblock \emph{Independent random sampling methods}.
\newblock Springer, 2018.

\bibitem[McSherry and Talwar(2007)]{mcsherry2007mechanism}
Frank McSherry and Kunal Talwar.
\newblock Mechanism design via differential privacy.
\newblock In \emph{48th Annual IEEE Symposium on Foundations of Computer
  Science (FOCS'07)}, pages 94--103. IEEE, 2007.

\bibitem[McSherry(2009)]{mcsherry2009privacy}
Frank~D McSherry.
\newblock Privacy integrated queries: an extensible platform for
  privacy-preserving data analysis.
\newblock In \emph{Proceedings of the 2009 ACM SIGMOD International Conference
  on Management of data}, pages 19--30, 2009.

\bibitem[Minami et~al.(2016)Minami, Arai, Sato, and
  Nakagawa]{minami2016differential}
Kentaro Minami, Hiromi Arai, Issei Sato, and Hiroshi Nakagawa.
\newblock Differential privacy without sensitivity.
\newblock In \emph{Advances in Neural Information Processing Systems}, pages
  956--964, 2016.

\bibitem[Mironov(2012)]{mironov2012significance}
Ilya Mironov.
\newblock On significance of the least significant bits for differential
  privacy.
\newblock In \emph{Proceedings of the 2012 ACM conference on Computer and
  communications security}, pages 650--661, 2012.

\bibitem[Mironov(2017)]{mironov2017Renyi}
Ilya Mironov.
\newblock {R{\'e}nyi} differential privacy.
\newblock In \emph{2017 IEEE 30th Computer Security Foundations Symposium
  (CSF)}, pages 263--275. IEEE, 2017.

\bibitem[Mohan et~al.(2012)Mohan, Thakurta, Shi, Song, and
  Culler]{mohan2012gupt}
Prashanth Mohan, Abhradeep Thakurta, Elaine Shi, Dawn Song, and David Culler.
\newblock {GUPT}: privacy preserving data analysis made easy.
\newblock In \emph{Proceedings of the 2012 ACM SIGMOD International Conference
  on Management of Data}, pages 349--360, 2012.

\bibitem[M{\o}ller and Nicholls(1999)]{moller1999perfect}
Jesper M{\o}ller and Geoff~K Nicholls.
\newblock \emph{Perfect simulation for sample-based inference}.
\newblock University of Aarhus. Centre for Mathematical Physics and
  Stochastics, 1999.

\bibitem[Nilizadeh et~al.(2019)Nilizadeh, Noller, and
  Pasareanu]{nilizadeh2019diffuzz}
Shirin Nilizadeh, Yannic Noller, and Corina~S Pasareanu.
\newblock Diffuzz: differential fuzzing for side-channel analysis.
\newblock In \emph{2019 IEEE/ACM 41st International Conference on Software
  Engineering (ICSE)}, pages 176--187. IEEE, 2019.

\bibitem[Pronzato and M{\"u}ller(2012)]{pronzato2012design}
Luc Pronzato and Werner~G M{\"u}ller.
\newblock Design of computer experiments: space filling and beyond.
\newblock \emph{Statistics and Computing}, 22\penalty0 (3):\penalty0 681--701,
  2012.

\bibitem[Propp and Wilson(1998)]{propp1998coupling}
James Propp and David Wilson.
\newblock Coupling from the past: a user’s guide.
\newblock \emph{Microsurveys in discrete probability}, 41:\penalty0 181--192,
  1998.

\bibitem[Reimherr and Awan(2019)]{reimherr2019kng}
Matthew Reimherr and Jordan Awan.
\newblock {KNG}: The k-norm gradient mechanism.
\newblock In \emph{Advances in Neural Information Processing Systems}, pages
  10208--10219, 2019.

\bibitem[Roy et~al.(2010)Roy, Setty, Kilzer, Shmatikov, and
  Witchel]{roy2010airavat}
Indrajit Roy, Srinath~TV Setty, Ann Kilzer, Vitaly Shmatikov, and Emmett
  Witchel.
\newblock Airavat: Security and privacy for mapreduce.
\newblock In \emph{NSDI}, volume~10, pages 297--312, 2010.

\bibitem[Seeman et~al.(2021)Seeman, Reimherr, and
  Slavkovi{\'c}]{seeman2021exact}
Jeremy Seeman, Matthew Reimherr, and Aleksandra Slavkovi{\'c}.
\newblock Exact privacy guarantees for markov chain implementations of the
  exponential mechanism with artificial atoms.
\newblock \emph{Advances in Neural Information Processing Systems},
  34:\penalty0 13125--13136, 2021.

\bibitem[Smith(2009)]{smith2009foundations}
Geoffrey Smith.
\newblock On the foundations of quantitative information flow.
\newblock In \emph{International Conference on Foundations of Software Science
  and Computational Structures}, pages 288--302. Springer, 2009.

\bibitem[Snoke and Slavkovi{\'c}(2018)]{snoke2018pmse}
Joshua Snoke and Aleksandra Slavkovi{\'c}.
\newblock {pMSE} mechanism: differentially private synthetic data with maximal
  distributional similarity.
\newblock In \emph{International Conference on Privacy in Statistical
  Databases}, pages 138--159. Springer, 2018.

\bibitem[Tang et~al.(2017)Tang, Korolova, Bai, Wang, and Wang]{tang2017privacy}
Jun Tang, Aleksandra Korolova, Xiaolong Bai, Xueqiang Wang, and Xiaofeng Wang.
\newblock Privacy loss in {Apple}'s implementation of differential privacy on
  macos 10.12.
\newblock \emph{arXiv preprint arXiv:1709.02753}, 2017.

\bibitem[Vempala and Wibisono(2019)]{vempala2019rapid}
Santosh Vempala and Andre Wibisono.
\newblock Rapid convergence of the unadjusted {Langevin} algorithm:
  Isoperimetry suffices.
\newblock In \emph{Advances in Neural Information Processing Systems}, pages
  8094--8106, 2019.

\bibitem[Wagh et~al.(2018)Wagh, Cuff, and Mittal]{wagh2018differentially}
Sameer Wagh, Paul Cuff, and Prateek Mittal.
\newblock Differentially private oblivious ram.
\newblock \emph{Proceedings on Privacy Enhancing Technologies}, 4:\penalty0
  64--84, 2018.

\bibitem[Wang et~al.(2019)Wang, Chen, and Xu]{wang2019differentially}
Di~Wang, Changyou Chen, and Jinhui Xu.
\newblock Differentially private empirical risk minimization with non-convex
  loss functions.
\newblock In \emph{International Conference on Machine Learning}, pages
  6526--6535. PMLR, 2019.

\bibitem[Wang et~al.(2015)Wang, Fienberg, and Smola]{wang2015privacy}
Yu-Xiang Wang, Stephen Fienberg, and Alex Smola.
\newblock Privacy for free: Posterior sampling and stochastic gradient {Monte
  Carlo}.
\newblock In \emph{International Conference on Machine Learning}, pages
  2493--2502, 2015.

\bibitem[Warner(1965)]{warner1965randomized}
Stanley~L Warner.
\newblock Randomized response: A survey technique for eliminating evasive
  answer bias.
\newblock \emph{Journal of the American Statistical Association}, 60\penalty0
  (309):\penalty0 63--69, 1965.

\bibitem[Wasserman and Zhou(2010)]{wasserman2010statistical}
Larry Wasserman and Shuheng Zhou.
\newblock A statistical framework for differential privacy.
\newblock \emph{Journal of the American Statistical Association}, 105\penalty0
  (489):\penalty0 375--389, 2010.

\bibitem[Zhang et~al.(2016)Zhang, Rubinstein, and
  Dimitrakakis]{zhang2016differential}
Zuhe Zhang, Benjamin~IP Rubinstein, and Christos Dimitrakakis.
\newblock On the differential privacy of {Bayesian} inference.
\newblock In \emph{Proceedings of the Thirtieth AAAI Conference on Artificial
  Intelligence}, pages 2365--2371, 2016.

\end{thebibliography}
\end{document}